\documentclass[a4paper, USenglish, thm-restate,pdfa]{lipics-v2021}

\usepackage{hyperref}
\usepackage{tikz}

\usepackage{todonotes}
\usepackage{mathtools}
\usepackage{etoolbox}
\usepackage{nicefrac}

\pdfoutput=1
\hideLIPIcs

\title{Computational complexity of the Weisfeiler-Leman dimension}

\author{Moritz Lichter}{RWTH Aachen University}{lichter@lics.rwth-aachen.de}{https://orcid.org/0000-0001-5437-8074}{The research of this author has received further funding by the European Union (ERC, SymSim, 101054974).}
\author{Simon Raßmann}{TU Darmstadt}{rassmann@mathematik.tu-darmstadt.de}{https://orcid.org/0000-0003-1685-410X}{}
\author{Pascal Schweitzer}{TU Darmstadt}{schweitzer@mathematik.tu-darmstadt.de}{https://orcid.org/0009-0001-3585-8213}{}
\authorrunning{M. Lichter, S. Raßmann and P. Schweitzer}
\Copyright{Moritz Lichter, Simon Raßmann and Pascal Schweitzer}



\funding{The research leading to these results has received funding from the European Research Council (ERC) under the European Union’s Horizon 2020 research and innovation programme (EngageS: grant agreement No.\ 820148).}

\nolinenumbers

\usepackage[utf8]{inputenc}
\usepackage[T1]{fontenc}
\usepackage{lmodern}
\usepackage{microtype}
\usepackage{anyfontsize}

\usepackage[l2tabu, orthodox, experimental]{nag}

\usepackage{todonotes}

\usepackage{hyperref}
\hypersetup{unicode}

\usepackage{bm}

\usepackage{xspace}

\usepackage{makeidx}

\makeatletter
\g@addto@macro\bfseries{\boldmath}
\makeatother

\tikzset{>=stealth}
\usepackage{amsmath}
\usepackage{amsfonts}
\usepackage{amssymb}
\usepackage{amsthm}
\usepackage{mathtools}
\usepackage{stmaryrd}
\usepackage{cancel}
\usepackage{tikz}
\usetikzlibrary{calc, cd, shapes}
\usepackage{pgf}

\SetSymbolFont{stmry}{bold}{U}{stmry}{m}{n}

\renewcommand*{\vec}[1]{{\mathbf{#1}}}
\newcommand{\eval}[1]{\llbracket #1\rrbracket}
\newcommand*{\bigunion}{\bigcup}

\newcommand*{\dotcup}{\mathbin{\dot\cup}}

\newcommand{\N}{\mathbb{N}}

\newcommand*{\C}{\mathrm{C}}

\newcommand{\True} {\textup{\textsc{True}}\xspace}
\newcommand{\False}{\textup{\textsc{False}}\xspace}

\newcommand*{\NP}{\textsf{NP}\xspace}
\newcommand*{\coNP}{\textsf{co-NP}\xspace}
\newcommand*{\Ptime}{\textsf{P}\xspace}
\newcommand*{\ACZ}{\textsf{AC\textsubscript{0}}\xspace}

\newcommand{\WL}[1][]{\ifstrempty{#1}{\ensuremath{\mathrm{WL}}}{\ensuremath{#1}\text{-}\ensuremath{\mathrm{WL}}}}
\DeclareMathOperator{\WLdim}{WL-dim}
\DeclareMathOperator{\tw}{tw}

\newcommand{\spl}{{\mathrm{split}}}

\DeclareMathOperator{\dom}{dom}

\DeclareMathOperator{\Aut}{Aut}
\DeclareMathOperator{\Sym}{Sym}

\DeclareMathOperator{\val}{val}

\DeclareMathOperator{\CFI}{CFI}

\newcommand*{\simeqq}{\cong}

\newcommand{\vertex}[3][black]{\node[shape=circle,draw=#1, fill=#1, inner sep=2.5pt] (#2) at (#3) {}}

\newcommand{\cconf}[1]{\mathcal{#1}} % coherent configurations
\newcommand{\rb}[1]{\mathcal{#1}} % rainbows
\newcommand{\CC}[1]{\operatorname{WL}_{#1}} % coherent closure
\newcommand{\F}[2][]{F_{#1}(#2)} % fibers
\newcommand{\Eq}[1]{\mathsf{Eq}_{#1}} % Eq-relation
\newcommand{\intnum}[2]{p\left(#1;#2\right)}%\operatorname{int}\kern-0.1em\left(#1;#2\right)} % intersection numbers
\newcommand{\extnum}[2]{p\left(#1;#2\right)}%\operatorname{ext}\kern-0.1em\left(#1;#2\right)} % extension numbers

\newcommand{\salgaut}[1]{\mathbb{A}(#1)} % strict algebraic automorphisms
\newcommand{\relstruc}[1]{\mathfrak{#1}} % relational structures
\newcommand{\feq}{\preceq}

\makeatletter
\let\@vareps\varepsilon
\let\varepsilon\epsilon
\let\epsilon\@vareps

\let\@varph\varphi
\let\varphi\phi
\let\phi\@varph
\makeatother

\begin{document}

\maketitle

\begin{abstract}
The Weisfeiler-Leman dimension of a graph \(G\) is the least number \(k\)
such that the \(k\)-dimensional Weisfeiler-Leman algorithm
distinguishes \(G\) from every other non-isomorphic graph, or equivalently,
the least \(k\) such that \(G\) is definable in \((k+1)\)-variable logic with counting.
The dimension is a standard measure of the descriptive or structural complexity of a graph
and recently finds various applications in particular in the context of machine learning.
This paper studies the complexity of computing the Weisfeiler-Leman dimension.
We observe that deciding whether the Weisfeiler-Leman dimension of $G$ is at most $k$
is \NP-hard, even if \(G\) is restricted to have \(4\)-bounded color classes.
Therefore, we study parameterized versions of the problem. 
For each fixed \(k\geq 2\), we give a polynomial-time algorithm
that decides whether the Weisfeiler-Leman dimension
of a given graph with \(5\)-bounded color classes is at most \(k\).
Moreover, we show that for these bounds on the color classes, this is optimal
because the problem is \Ptime-hard under logspace-uniform \ACZ-reductions.
Furthermore, for each larger bound $c$ on the color classes and each fixed $k \geq 2$,
we provide a polynomial-time decision algorithm for the abelian case,
that is, for structures of which each color class has an abelian automorphism group.

While the graph classes we consider may seem quite restrictive, 
graphs with \(4\)-bounded abelian colors include CFI-graphs and multipedes,
which form the basis of almost all known hard instances and lower bounds related to the Weisfeiler-Leman algorithm.
\end{abstract}

\section{Introduction}
The \emph{Weisfeiler-Leman algorithm}
is a simple combinatorial procedure studied in the context of the graph isomorphism problem.
For every \(k\geq 1\), the algorithm has a \(k\)-dimensional variant, \WL[k] for short,
that colors \(k\)-tuples of vertices according to how they structurally
sit inside the whole graph:
if two tuples get different colors, they cannot be mapped onto each other by an automorphism of the graph (while the converse is not always true).
The \(1\)-dimensional algorithm, which is also known as \emph{color refinement},
starts by coloring each vertex according to its degree, and then
repeatedly refines this coloring by including into each vertex color
the multisets of colors of its neighbors.
The \(k\)-dimensional variant generalizes this idea and colors \(k\)-tuples of vertices
instead of single vertices~\cite{WL,CFI}.

The Weisfeiler-Leman algorithm plays an important role in both theoretic and practical
approaches to the graph isomorphism problem, but is also related to a plethora of seemingly
unrelated areas: to finite model theory and descriptive complexity via the correspondence of \WL[k] to \((k+1)\)-variable
first-order logic with counting~\cite{CFI, ImmermanLanderCCS3},
to machine learning via a correspondence to the expressive power of (higher-dimensional)
graph neural networks~\cite{WLgoesNeural},
to the Sherali-Adams hierarchy in combinatorial optimization \cite{sherali_adams1,sherali_adams2},
and to homomorphism counts from treewidth-\(k\) graphs \cite{wl_vs_homomorphisms}.

On the side of practical graph isomorphism, the color refinement procedure is a basic building block
of the so-called \emph{individualization-refinement framework},
which is the basis of almost every modern practical solver for the graph
isomorphism problem \cite{nauty, bliss1, bliss2, dejavu}.
On the side of theoretical graph isomorphism, Babai's quasipolynomial-time
algorithm for the graph isomorphism problem \cite{Quasipolynomial}
uses a combination of group-theoretic techniques
and a logarithmic-dimensional Weisfeiler-Leman algorithm.

The Weisfeiler-Leman algorithm is a powerful algorithm for distinguishing non-isomorphic graphs on its own.
For every $k$, \WL[k] can be used as an incomplete polynomial-time isomorphism test:
if the multiset of colors of \(k\)-tuples of two graphs $G$ and $H$ differ,
then $G$ and $H$ cannot be isomorphic.
In this case, \WL[k] distinguishes $G$ and $H$, otherwise $G$ and $H$ are \WL[k]-equivalent.
For a given graph \(G\), we say that the \(k\)-dimensional Weisfeiler-Leman algorithm \WL[k] \emph{identifies} \(G\) if
it distinguishes \(G\) from every non-isomorphic graph.
The smallest such \(k\) is known as the \emph{Weisfeiler-Leman dimension of~\(G\)}~\cite{GroheMinors}.

It is known that almost all graphs have Weisfeiler-Leman dimension~\(1\)~\cite{CRIdentifiesAlmostAllGraphs}.
However, color refinement fails spectacularly on regular graphs, where it always returns the
monochromatic coloring. For these, it is known that \WL[2] identifies almost all regular graphs \cite{2WLIdentifiesAlmostAllRegularGraphs, KuceraRandomRegularGraphs}.
In contrast to these positive results,
for every~\(k\) there is some graph \(G\)
(even of order linear in~$k$, of maximum degree~\(3\), and with \(4\)-bounded abelian color classes,
i.e., such that no more than~\(4\) vertices can share the same vertex-color
and every color class induces a graph with abelian automorphism group)
that is not identified by \WL[k]~\cite{CFI}. These so-called CFI-graphs
have high Weisfeiler-Leman dimension
and are thus hard instances for combinatorial approaches to the graph isomorphism problem.

The situation changes for restricted classes of graphs.
If the Weisfeiler-Leman dimension over some class of graphs is bounded by \(k\),
then the \(k\)-dimensional Weisfeiler-Leman algorithm correctly decides isomorphism over this class.
And since \WL[k] can be implemented in polynomial time \(O(n^{k+1}\log n)\) \cite{ImmermanLanderCCS3},
this puts graph isomorphism over such classes into polynomial time.
Examples of graph classes with bounded Weisfeiler-Leman dimension include graphs of bounded tree-width~\cite{BoundedTreeWidth}, graphs of bounded rank-width \cite{BoundedRankWidth},
graphs with \(3\)-bounded color classes.
\cite{ImmermanLanderCCS3},
planar graphs \cite{WL3onPlanarGraphs}, and more generally
every non-trivial minor-closed graph class \cite{GroheMinors}.

In this paper, we study the computational complexity of computing the Weisfeiler-Leman dimension. We call the problem of deciding whether the Weisfeiler-Leman dimension of a given graph
is at most \(k\) the \emph{\WL[k]-identification problem}.
For upper complexity bounds, non-identification of a graph \(G\) can be witnessed
by providing a graph \(H\) that is not distinguished from \(G\)
by \WL[k] but is also not isomorphic to \(G\).
As the latter can be checked in \coNP, this places the identification
problem into the class \(\Pi_2^{\Ptime}\) of the polynomial hierarchy.
If the graph isomorphism problem is solvable in polynomial time,
this complexity bound collapses to \coNP.
However, there is no apparent reason why the identification
problem should not be polynomial-time decidable.

On the side of lower complexity bounds, the \WL[1]-identification problem is
complete for polynomial time under uniform reductions in the circuit complexity class \ACZ \cite{CRIdentification,CRIdentificationPHard}.
Hardness of the \WL[1]-identification problem does, however, not easily imply any hardness results for the \WL[k]-identification problem for higher values of \(k\).
Indeed, no hardness results are known for \(k\geq 2\).
The \WL[2]-identification problem in particular includes
the problem of deciding whether a given strongly regular graph is determined up to isomorphism
by its parameters, which is a baffling problem from classic combinatorics far beyond our current knowledge. 
To understand the difficulties of the \WL[k]-identification problem better, we can again consider classes of graphs.
On every class of graphs with bounded color classes, graph isomorphism
is solvable in polynomial time \cite{GIBoundedCCSLasVegas,GIBoundedCCS},
which puts the identification problem over this class into \coNP for every \(k\geq 2\).
Graphs with \(3\)-bounded color classes
are identified by \WL[2]~\cite{ImmermanLanderCCS3}, which makes their identification problem trivial.
As shown by the CFI-graphs~\cite{CFI}, this is no longer true for graphs with \(4\)-bounded color classes.
Nevertheless, as shown by Fuhlbrück, Köbler, and Verbitsky, identification of graphs
with \(5\)-bounded color classes by \WL[2] is efficiently decidable~\cite{WL2onCCS4}.
For higher dimensions or bounds on the color classes essentially nothing is known.

\subparagraph*{Contribution.}
We extend the results of~\cite{WL2onCCS4} from~\WL[2] to~\WL[k] and
give a polynomial-time algorithm
deciding 
whether a graph with \(5\)-bounded color classes
is identified by \(\WL[k]\):
\begin{theorem}[name=,restate=WLkIDccsFive]\label{Theorem: WL-identification on ccs 5}
For every $k$, there is an algorithm that decides the
\WL[k]-identification problem for vertex- and edge-colored,
directed graphs with $5$-bounded color classes in time \(O_{k}(n^{O(k)})\).
If such a graph \(G\) is not identified by \(\WL[k]\),
the algorithm provides a witness for this, i.e., a graph \(H\) that is not isomorphic to \(G\) and not distinguished from \(G\) by \WL[k].
\end{theorem}

\noindent Via the correspondence of \WL[k] to $(k+1)$-variable counting logic,
Theorem~\ref{Theorem: WL-identification on ccs 5} implies that definability of graphs with $5$-bounded color classes in this logic is decidable in polynomial time.
While the restriction to \(5\)-bounded color classes may seem stark,
almost all known hardness results and lower bounds for the Weisfeiler-Leman algorithm remain true
for graphs with bounded color classes and in most cases even \(4\)-bounded color classes suffice~\cite{WLEquivalencePHard, CFIvstw, schneider_upperBound, multipedesI, multipedesII}.

Towards generalizing Theorem~\ref{Theorem: WL-identification on ccs 5} to arbitrary relational structures
and larger color classes, we consider structures with abelian color classes,
i.e., structures of which each color class induces a structure with an abelian automorphism group.
Such structures were previously considered in the context of descriptive complexity theory \cite{abelian_color_classes},
and include both CFI-graphs \cite{CFI} and multipedes \cite{multipedesI, multipedesII} over ordered base graphs, which form the basis of all known constructions of graphs with high Weisfeiler-Leman dimension.
For many case in descriptive  complexity theory,
restricting to \(4\)-bounded abelian color classes is sufficient,
but in some cases larger (but still abelian) color classes are required~\cite{Bijective_Pebble_Games,GradelPakusa19, Lichter2023, Lichter2023b}.
For such structures, we obtain a polynomial-time algorithm as before:
\begin{theorem}[restate=IdentificationAbelianColors, name=]\label{thm:abelian_colors:deciding_identication}
For every $k\in\N$ and \(c,r\leq k\), there is an algorithm that decides the
\WL[k]-identification problem for
\(r\)-ary relational structures with \(c\)-bounded abelian color classes in time \(O_{k}(n^{O(k)})\).
If such a structure \(\relstruc{A}\) is not identified by \(\WL[k]\),
the algorithm provides a witness for this, i.e., a second structure \(\relstruc{B}\) that is not isomorphic to \(\relstruc{A}\) and not distinguished from \(\relstruc{A}\) by \WL[k].
\end{theorem}

\noindent On the side of hardness results, we first prove that when the dimension \(k\) is part of the input, the identification problem is \NP-hard.
Note that a similar result was recently independently observed by Seppelt~\cite{seppelt_WLEquivcoNPhard}.
\begin{theorem}[name=,restate=WLdimNPhard]\label{Theorem: Computing WL-dimension NP-hard}
The problem of deciding,
given a graph \(G\) and a natural number \(k\), 
whether the Weisfeiler-Leman dimension of $G$ is at most $k$
is \NP-hard, both over uncolored simple graphs, and over simple graphs with \(4\)-bounded abelian color classes.
\end{theorem}
\noindent Furthermore, we extend the \Ptime-hardness results for \WL[1] \cite{CRIdentificationPHard}
to arbitrary \(k\) and prove that, when \(k\) is fixed, the \WL[k]-identification problem
is hard for polynomial time:
\begin{theorem}[name=,restate=IdentificationPHard]\label{Theorem: WL-identification P-hard}
For every \(k\geq 1\), the \WL[k]-identification problem is \Ptime-hard under uniform \ACZ-reductions
over both uncolored simple graphs, and simple graphs with \(4\)-bounded abelian color classes.
\end{theorem}

\subparagraph*{Techniques.}
To prove Theorem~\ref{Theorem: WL-identification on ccs 5},
we exploit the close connection between the coloring computed by \WL[k]
and \(k\)-ary coherent configurations.
These structures come with two notions of isomorphisms,
algebraic ones and combinatorial ones.
Similarly to~\cite{WL2onCCS4}, we reduce the \WL[k]-identification problem to the separability
problem for \(k\)-ary coherent configurations, 
that is, to decide whether algebraic and combinatorial isomorphisms for a given \(k\)-ary coherent configuration coincide.
We make two crucial observations:
First, we show that the \(k\)-ary coherent configurations obtained from graphs
are fully determined by their underlying \(2\)-ary configurations.
We call such configurations \(2\)-induced.
Second, we reduce the separability problem for arbitrary \(k\)-ary coherent configurations
to that of \(k\)-ary coherent configurations where no interspace contains a disjoint union of stars.
Combining both observations,
we show that two \(2\)-induced, star-free \(k\)\nobreakdash-ary coherent configurations
obtained from \WL[k]-equivalent graphs must be isomorphic.
Given such a \(k\)-ary coherent configuration obtained from a graph,
it thus suffices to decide whether
there is another non-isomorphic graph yielding the same configuration.
Finally, we solve this problem by
encoding it into the graph isomorphism problem
for structures with bounded color classes,
which is polynomial-time solvable~\cite{GIBoundedCCSLasVegas, GIBoundedCCS}.

The main obstacle to generalize Theorem~\ref{Theorem: WL-identification on ccs 5} to larger color classes or relational structures of higher arity is the existence of \WL[k]-equivalent structures
that yield non-isomorphic star-free \(k\)-ary coherent configurations,
which greatly increases the space of possibly equivalent bot non-isomorphic structures.

To make up for this, we consider structures with abelian color classes.
Using both the bijective pebble game \cite{Bijective_Pebble_Games} and ideas
from the theory of coherent configurations, we provide structural insights for
the class of \(k\)-ary coherent configurations with abelian fibers
which allows us to finally prove that in this case, it does suffice
to consider other relational structures yielding the same
\(k\)-ary coherent configuration.

\NP-hardness in Theorem~\ref{Theorem: Computing WL-dimension NP-hard}
is proved by combining the known relationship between the Weisfeiler-Leman dimension
of CFI-graphs \cite{CFI} and the tree-width of the underlying base graphs
with the  recent result that computing the tree-width of cubic graphs
is \NP-hard~\cite{twNPhard}. With the same techniques, we can also prove that
deciding \WL[k]-equivalence of graphs is \coNP-hard when the dimension \(k\) is considered
part of the input.

For the \Ptime-hardness result of the \WL[k]-identification problem in Theorem~\ref{Theorem: WL-identification P-hard},
we adapt a construction by Grohe \cite{WLEquivalencePHard} to encode
monotone boolean circuits into graphs using different types of gadgets.
This simultaneously reduces the monotone circuit value problem, which is known to be hard for polynomial time,
to the \WL[k]-equivalence and \WL[k]-identification problem.
The main difficulty was showing identification of Grohe's gadgets, specifically his so-called \emph{one-way switches}.
We give an alternative construction of these one-way switches based on the
CFI-construction \cite{CFI}.
This construction simplifies proofs and more importantly
yields graphs with \(4\)-bounded abelian color classes for every $k$.
This shows hardness for the \WL[k]-equivalence
and \WL[k]-identification problems even for graphs with \(4\)-bounded abelian color classes.

\section{The Weifeiler-Leman algorithm and coherent configurations}\label{sec:preliminaries}

\paragraph*{Preliminaries.}
For \(n\in\N\), we set \([n] \coloneqq \{1,\dots,n\}\).
For a set~$A$, the set of all \(k\)-element subsets of~$A$ is denoted by \(\binom{A}{k}\).
For two runtime-bounding functions \(f\) and \(g\) with parameters including \(\kappa\),
we write \(f\in O_\kappa(g)\) if \(f/g\) is bounded by a function of \(\kappa\).
A simple graph is a pair \(G=(V(G),E(G))\) of a set \(V(G)\) of \emph{vertices}
and a set \(E(G)\subseteq\binom{V(G)}{2}\) of undirected edges.
For a directed graph, we allow \(E(G)\subseteq V(G)^2\setminus\{(v,v)\colon v\in V(G)\}\).
For either graph type, we write \(uv\) for the edge \(\{u,v\}\) or \((u,v)\) respectively.
For a simple or directed graph \(G\), a \emph{vertex-coloring} of \(G\) is a map
\(\chi\colon V(G)\to C\) for some finite, ordered set \(C\) of colors.
Similarly, an \emph{edge-coloring} is a map \(\eta\colon E(G)\to C\).
A (vertex-)color class is a set \(\chi^{-1}(c)\) for some vertex color \(c\in C\).
If all color classes have order at most \(q\), we say that the colored graph \((G,\chi)\)
has \emph{\(q\)-bounded color classes}.

Relational structures are a higher-arity analogue of graphs.
Formally, a \emph{\(k\)-ary relational structure} \(\relstruc{A}\)
is a tuple \((V(\relstruc{A}),R_1,\dots,R_\ell)\) of vertices \(V(\relstruc{A})\)
and relations \(R_i\subseteq V(\relstruc{A})^{r_i}\) with \(r_i\leq k\).
The number \(r_i\) is the \emph{arity} of the relation \(R_i\).
We again allow relational structures to come with a vertex-coloring and define
\(q\)-bounded color classes as before.

An \emph{isomorphism} between graphs \(G\) and \(H\) is a bijection
\(\phi\colon V(G)\to V(H)\) such that \(uv\in E(G)\) if and only if \(\phi(u)\phi(v)\in E(H)\).
In this case \(G\) and \(H\) are \emph{isomorphic} and we write \(G\simeqq H\).
An isomorphism between edge- or vertex-colored graphs must also preserve the vertex- and edge-colors.
Similarly, an isomorphism between (vertex-colored) relational structures is a (color-preserving)
bijection between the vertex sets that preserves all relations and their complements.
An automorphism is an isomorphism from a structure to itself. The set of automorphisms
of a structure forms a group, and we say that a graph or relational structure \(\relstruc{A}\)
has \emph{abelian color classes} if for every color class \(C\), the induced substructure
\(\relstruc{A}[C]\) has an abelian automorphism group.
\paragraph*{The Weisfeiler-Leman algorithm.}
For every \(k\geq 2\), the \(k\)-dimensional Weisfeiler-Leman algorithm (\WL[k])
computes an isomorphism-invariant coloring of \(k\)-tuples of vertices
of a given graph~$G$ via an iterative refinement process.
Initially, the algorithm colors each \(k\)-tuple
according to its \emph{isomorphism type},
i.e., $\vec{x}=(x_1,\dots,x_k),\vec{y}=(y_1,\dots,y_k)\in V(G)^k$
get the same color if and only if mapping $x_i \mapsto y_i$ for every $i\in[k]$ is an isomorphism of
the induced subgraphs $G[\{x_1,\dots,x_k\}]$ and $G[\{y_1,\dots, y_k\}]$.
In each iteration, this coloring is refined as follows:
if \(\chi^G_r \colon V(G)^k \to C_i\) is the coloring obtained after \(i\) refinement rounds,
the coloring \(\chi^G_{r+1}\colon V(G)^k \to C_{i+1}\) is defined as
\(\chi^G_{r+1}(\vec{x})\coloneqq(\chi^G_r(\vec{x}),M_\vec{x}^i)\), where
\[ M_\vec{x}^r=\left\{\!\!\left\{
\Bigl(\chi^G_r\bigl(\vec{x}\frac{y}{1}\bigr),
      \dots,
      \chi^G_r\bigl(\vec{x}\frac{y}{k}\bigr)\Bigr)
\colon y\in V(G)
\right\}\!\!\right\}\]
and $\vec{x}\frac{y}{i}$ denotes the tuple obtained from $\vec{x}$
by replacing the $i$-th entry by $y$.
If \(\chi^G_{r+1}\) does not induce a finer color partition on $V(G)^k$ than \(\chi^G_r\),
the algorithm terminates and returns the stable coloring~\(\chi^G_\infty \coloneqq \chi^G_r\).
This must happen before the \(n^k\)-th refinement round.

We say that \WL[k] \emph{distinguishes
two \(k\)-tuples $\vec{x},\vec{y} \in V(G)^k$} if $\chi^G_\infty(\vec{x}) \neq \chi^G_\infty(\vec{y})$
and that \WL[k] \emph{distinguishes two \(\ell\)-tuples $\vec{x},\vec{y} \in V(G)^\ell$} for \(\ell<k\)
if \WL[k] distinguishes the two \(k\)-tuples we get by repeating the last entries of $\vec{x}$ respectively $\vec{y}$.
Finally, \emph{\WL[k] distinguishes two graphs}~$G$ and~$H$
if there is a color $c$ such that
\[\left|\left\{\vec{x} \in V(G) : \chi_\infty^G(\vec{x}) = c\right\}\right| \neq \left|\left\{\vec{x} \in V(H) \colon \chi_\infty^H(\vec{x}) = c\right\}\right|.\]
Otherwise,~$G$ and~$H$ are \emph{\WL[k]-equivalent} and we write $G \equiv_{\WL[k]} H$.
A graph~$G$ is \emph{identified} by \WL[k]
if \WL[k] distinguishes~$G$ from every other non-isomorphic graph.
Every \(n\)-vertex graph is identified by \WL[n], and the least number~\(k\) such that \WL[k] identifies~\(G\) is called the \emph{Weisfeiler-Leman dimension} of~\(G\),
denoted by \(\WLdim(G)\).

If \(c_1,\dots,c_k\) are colors assigned by the stable coloring \(\chi^G_r\),
then the multiplicity of the tuple \((c_1,\dots,c_k)\) in the multiset
\(M^r_{\vec{x}}\) is given by
\[p^\vec{x}_{c_1,\dots,c_k}\coloneqq\left|\left\{y\in V(G)\colon
	\chi^G_r\left(\vec{x}\frac{y}{i}\right)=c_i\text{ for all } i\in[k]\right\}\right|.\]
These numbers, which are called \emph{intersection numbers}, fully determine
the multiset \(M^r_{\vec{x}}\).
In the theoretical study of the Weisfeiler-Leman algorithm,
these numbers are often quite handy to work with compared to the multiset
view from the definition.
For example, the statement that a coloring \(\chi\) is stable under \WL[k]-refinement
can be expressed as the intersection numbers
\(p^{\vec{x}}_{c_1,\dots,c_k}\) being determined by \(\chi(\vec{x})\)
and not depending on \(\vec{x}\) itself.
This leads to the notion of \emph{coherent configurations}.

As every coloring of \(k\)-tuples also induces a coloring of \(\ell\)-tuples for 
\(\ell\leq k\) by repeating the last entry, \WL[k] is at least
as powerful in distinguishing graphs as \WL[\ell] and
this hierarchy is actually strict~\cite{CFI}.
Completely analogously, \WL[k] can be applied to relational structures.
\paragraph*{Coherent configurations.}
For an introduction to (\(2\)-ary) coherent configurations we refer to~\cite{CC}
and for their connection to the Weisfeiler-Leman algorithm we refer to \cite{WL2onCCS4}.
For \(k\geq 2\), a \emph{\(k\)-ary rainbow} is a pair \((V,\rb{R})\)
of a finite set of vertices \(V\) and a partition
\(\rb{R}\) of~\(V^k\), whose elements are called \emph{basis relations},
that satisfies the following two conditions:
\begin{description}
\item[(R1)] For every basis relations \(R\in\rb{R}\), all tuples \(\vec{x},\vec{y}\in R\)
	have the same \emph{equality type}, i.e., \(x_i=x_j\) if and only if \(y_i=y_j\).
	We also call this the equality type of the relation \(R\).
\item[(R2)] \(\rb{R}\) is closed under permuting indices:
	For all basis relations \(R\in\rb{R}\)
	and permutations \(\sigma\)
	of~\([k]\),
	the set \(R^\sigma \coloneqq  \{(x_{\sigma(1)}, \dots, x_{\sigma(k)}) \colon (x_1,\dots,x_k) \in R\}\) is a basis relation.
\end{description}
Because the vertex set \(V\) is determined by the partition \(\rb{R}\), we also write
\(\rb{R}\) to denote the rainbow \((V,\rb{R})\) and in this case write \(V(\rb{R})\) for
its vertex set \(V\).

A \emph{\(k\)-ary coherent configuration} is a \(k\)-ary rainbow \(\cconf{C}\)
that is stable under \WL[k]-refinement. More formally, this means that
\begin{description}
\item[(C)] for all basis relations \(R,R_1,\dots,R_k\in\cconf{C}\),
	the \emph{intersection number}
	\[\intnum{R}{R_1,\dots,R_k}\coloneqq\left|\left\{y\in V(\cconf{C})\colon \vec{x}\frac{y}{i}\in R_i \text{ for all } i\in[k]\right\}\right|\]
	is the same for all choices of \(\vec{x}\in R\) and is thus well-defined.
\end{description}
For \(\ell\leq k\), the partition of $k$-vertex tuples of an \(\ell\)-ary relational structure
according to their isomorphism type always yields a $k$-ary rainbow.
The connection of \WL[k] and $k$-ary coherent configurations is
that the partition of \(k\)-vertex tuples of a graph according to their \WL[k]-colors
always forms a \(k\)-ary coherent configuration.

\subparagraph*{Induced configurations.}
% induced partitions
If \(\rb{R}\) is an \(\ell\)-ary rainbow for \(\ell\leq k\), we can 
interpret \(\rb{R}\) as the \(k\)-ary rainbow \(\rb{R}|^k\) 
by partitioning \(k\)-tuples according
to the basis relations of the \(\ell\)-subtuples they contain.
Formally, let \(\sim_\rb{R}\) be the equivalence relation on \(V(\rb{R})^\ell\)
whose equivalence classes are the basis relations of \(\rb{R}\).
We define the equivalence relation \(\sim_\rb{R}^k\) on \(V(\rb{R})^k\) by writing
$\vec{x}\sim_k\vec{y}$ if and only if for all \(I\in\binom{[k]}{\ell}\) we have $\vec{x}|_I\sim_\rb{R}\vec{y}|_I$,
where $\vec{x}|_I$ is the subtuple of $\vec{x}$
for which all indices not in $I$ are deleted.
The basis relations of  \(\rb{R}|^k\) are the equivalence classes of \(\sim_\rb{R}^k\).

% Coherent closures
For every $k$-ary rainbow \(\rb{R}\), there is a unique coarsest \(k\)-ary coherent configuration  \(\CC{k}(\rb{R})\)
that is at least as fine as \(\rb{R}\) and is called the \emph{\(k\)-ary coherent closure}
of \(\rb{R}\).
For an \(\ell\leq k\) and an \(\ell\)-ary rainbow \(\rb{R}\), we also write \(\CC{k}(\rb{R})\) for \(\CC{k}(\rb{R}|^k)\).
Similarly, for an \(\ell\)-ary relational structure \(\relstruc{A}\), we write \(\CC{k}(\relstruc{A})\)
for the partition of \(V(\relstruc{A})^k\) into \(\WL[k]\)-color classes.

% Skeletons, faces, extension numbers
Every \(k\)-ary coherent configuration \(\cconf{C}\) induces the \(\ell\)-ary coherent configuration
\(\cconf{C}|_\ell\) for every \(\ell\leq k\) by considering the partition of tuples of the form
\((x_1,\dots,x_\ell,\dots,x_\ell)\in V(\cconf{C})^k\). This \(\ell\)-ary coherent configuration is called the \emph{\(\ell\)-skeleton of \(\cconf{C}\)}.
For every basis relation \(R\in\cconf{C}\) and every subset \(I\in\binom{[k]}{\ell}\) of the indices,
the set \(R_I\coloneqq\{\vec{x}|_I\colon\vec{x}\in R\}\) is a basis relation of~\(\cconf{C}|_{\ell}\)
and called the \emph{\(I\)-face} of \(R\).
For basis relations \(R\in\cconf{C}|_\ell\) and \(T\in\cconf{C}\),
we get well-defined \emph{extension numbers}
$\extnum{R}{T}\coloneqq|\{\vec{y}\in V(\cconf{C})^{k-\ell}\colon \vec{xy}\in T\}|$ for some (and every) $\vec{x}\in R$.

% fibers, interspaces, \ell-inducedness, colored variants
For \(\ell=1\), the \(1\)-skeleton yields a partition of \(V(\cconf{C})\),
whose partition classes are called \emph{fibers}. We denote the set of fibers by \(\F{\cconf{C}}\).
\(\cconf{C}\) has \emph{\(c\)-bounded fibers} if all fibers of \(\cconf{C}\) have order at most \(c\).
Between two fibers \(X\) and \(Y\), the induced configuration \(\cconf{C}|_2\)
further induces a partition \(\cconf{C}|_2[X,Y]\) of \(X\times Y\), called an \emph{interspace}.

We call a \(k\)-ary coherent configuration \(\cconf{C}\) \emph{\(\ell\)-induced} if it is the coherent closure of its \(\ell\)-skeleton,
i.e., if \(\cconf{C}=\CC{k}(\cconf{C}|_\ell)\). This is equivalent to \(\cconf{C}\) being the coherent closure
of some \(\ell\)-ary rainbow. In particular, the \(k\)-ary coherent closure of a (directed, colored) graph is \(2\)-induced and, more generally, the $k$-ary coherent closure
of an \(\ell\)-ary relational structure is \(\ell\)-induced for every \(k\geq \ell\).

For a \(k\)-ary rainbow \(\rb{R}=(V,\{R_1,\dots,R_\ell\})\),
the vertex-colored \(k\)-ary relational structure \((V,R_1,\dots,R_\ell,\chi)\)
where \(\chi\) maps every vertex to its fiber is
a \emph{colored variant} of \(\rb{R}\). Note that this requires choosing an ordering
of the basis relations; colored variants are thus not unique.

\subparagraph{Algebraic and combinatorial isomorphisms.}
There are two notions of isomorphism for two \(k\)-ary coherent configurations $\cconf{C}$ and $\cconf{D}$.
First, a \emph{combinatorial isomorphism} is a bijection \(\phi\colon V(\cconf{C})\to V(\cconf{D})\) that
preserves the partition into basis relations, i.e.,
for every basis relation \(R\in\cconf{C}\),
the mapped set $R^\phi \coloneqq \{ (\phi(x_1),\dots,\phi(x_k)) \colon (x_1,\dots,x_k) \in R\}$ is a basis relation of \(\cconf{D}\).
Combinatorial isomorphisms are thus isomorphisms between certain colored variants
of \(\cconf{C}\) and \(\cconf{D}\) and the notion also applies to rainbows.

Second, an \emph{algebraic isomorphism} is a map
\(f\colon\cconf{C}\to\cconf{D}\) between the two partitions
that preserves the intersection numbers. More formally, we require that
\begin{description}
\item[(A1)] for all \(R\in\cconf{C}\), the relations \(R\) and \(f(R)\) have the same equality type,
\item[(A2)] for all \(R\in\cconf{C}\) and permutations \(\sigma\) of $[k]$,
	we have \(f(R^\sigma)=f(R)^\sigma\), and
\item[(A3)] for all \(R,T_1,\dots,T_k\in\cconf{C}\), we have $\intnum{R}{T_1,\dots,T_k}=\intnum{f(R)}{f(T_1),\dots,f(T_k)}$,
\end{description}
but Property~(A3) already implies the former two.
Algebraic isomorphisms can be thought of as maps preserving the Weisfeiler-Leman colors
and thus as a functional perspective on Weisfeiler-Leman equivalence.
More formally, if for \(k\)-ary relational structures \(\relstruc{A}\) and \(\relstruc{B}\),
\(f\colon\CC{k}(\relstruc{A})\to\CC{k}(\relstruc{B})\) is an algebraic isomorphism
that preserves the relations of \(\relstruc{A}\) and \(\relstruc{B}\),
then \(f\) is the unique map that maps every color class of the stable coloring computed by \WL[k]
on \(\relstruc{A}\) to the corresponding color class of the stable coloring computed by \WL[k]
on \(\relstruc{B}\). In particular, we get \(\relstruc{A}\equiv_{\WL[k]}\relstruc{B}\) in this case.

If \(f\colon\cconf{C}\to\cconf{D}\) is an algebraic isomorphism, then~\(f\)
induces an algebraic isomorphism \(f|_\ell\colon\cconf{C}|_\ell\to\cconf{D}|_\ell\) for every $\ell \leq k$.
A combinatorial (respectively algebraic) \emph{automorphism of \(\cconf{C}\)} is a combinatorial
(respectively algebraic) isomorphism from~\(\cconf{C}\) to itself.
Every combinatorial isomorphism induces an algebraic isomorphism, but the converse is not true.
Algebraic isomorphisms behave nicely with coherent closures as seen in the next lemma
(the proof is analogue to the $k=2$ case~\cite[Lemma 2.4]{WL2onCCS4}):
\begin{lemma}
\label{Lemma: algebraic isomorphism and coherent closures}
Let \(\rb{R}\) be a \(k\)-ary rainbow, \(\cconf{C}=\CC{k}(\rb{R})\), and \(f\colon\cconf{C}\to\cconf{D}\) an algebraic isomorphism. Then
\begin{enumerate}
\item \(\cconf{D}=\CC{k}(\rb{R}^f)\),
	in particular, if \(\cconf{C}\) is \(\ell\)-induced, then so is \(\cconf{D}\),
\item \(f\) is fully determined by its action on basis relations in \(\rb{R}\), and
\item if \(f|_\rb{R}\) is induced by a combinatorial isomorphism \(\phi\),
	then \(\phi\) induces \(f\).
\end{enumerate}
\end{lemma}

A \(k\)-ary coherent configuration \(\cconf{C}\) is called \emph{separable}
if every algebraic isomorphism \(f\colon\cconf{C}\to\cconf{D}\) from \(\cconf{C}\)
is induced by a combinatorial one.
There is a close relation to the power of the Weisfeiler-Leman algorithm
(the proof is analogue to the $k=2$ case~\cite[Theorem 2.5]{WL2onCCS4}):
\begin{lemma}\label{Lemma: identification and separability}
Let \(\ell\leq k\) and \(\relstruc{A}\) be an \(\ell\)-ary relational structure. Then \(\relstruc{A}\) is identified by the \(k\)-dimensional Weisfeiler-Leman algorithm if and only if  \(\CC{k}(\relstruc{A})\) is separable.
\end{lemma}

\paragraph*{Bounded variable counting logics}
The \(k\)-dimensional Weisfeiler-Leman algorithm has an alternative characterization in terms of the distinguishing power
of some logic, namely \((k+1)\)-variable counting logic.
First-order counting logic \(\C\) is the extension of first-order logic by the \emph{counting quantifiers} \(\exists^{\geq k}\)
for all natural numbers \(k\), which state that there exist at least \(k\) distinct elements satisfying the formula that follows.
But because first-order logic has the ability to simulate the counting quantifier \(\exists^{\geq k}\) by a sequence of \(k\) usual
existential quantifiers, adding counting quantifiers does not actually increase the expressive power of first-order logic.
This situation changes when we restrict the number of variables.
For a natural number \(k\geq 2\), we define \(k\)-variable counting logic \(\C^k\) to be the fragment of \(\C\)
which only uses the variables \(x_1,\dots,x_k\). In order to not restrict the expressive power of these logics too much,
we do, however, allow \emph{requantifications}, that is, quantifications over a variable within the scope of another
quantification over the same variable. As an example, the following is a \(\C^2\)-formula
stating that
\[\forall x_1\exists x_2 \left(Ex_1x_2\land \left(\exists^{\geq 5} x_1 Ex_2x_1\right) \land \neg\exists^{\geq 6} x_1 Ex_2x_1\right),\]
which states that every vertex is adjacent to a vertex of degree \(5\).
\paragraph*{The bijective pebble game.}
The question whether \WL[k] can distinguish two graphs
\(G\) and \(H\) has another characterization in terms of the so-called \emph{bijective \((k+1)\)-pebble game}.
In this game, there are two players: Spoiler and Duplicator.
Game positions are partial maps \(\vec{g}\mapsto\vec{h}\) between \(G\) and \(H\),
where both tuples contain at most \(k+1\) elements. We also sometimes
identify such partial maps with the set \(P=\{g_i\mapsto h_i\colon i\leq|\vec{g}|\}\).
We think of these maps
as \(k+1\) pairs of corresponding pebbles placed in the two graphs.

If such a partial map is not a partial isomorphism,
i.e., not an isomorphisms on the induced subgraphs, Spoiler wins immediately.
Otherwise, at the beginning of each turn, Spoiler picks up one pebble pair, either from the board
if all \(k+1\) pairs are placed, or from the side if there are pebble pairs left. Duplicator
responds by giving a bijection \(\phi\colon V(G)\to V(H)\) between the two graphs.
Spoiler then places the pebble pair they picked up on a pair \((g,\phi(g))\) of vertices of their choice.
The game then continues in the resulting new position.

We say that Spoiler wins if the graphs have differing cardinality or
they can reach a position that is no longer a partial isomorphism (and thus win immediately).
Duplicator wins the game if they can find responses to Spoiler's moves indefinitely.

\begin{lemma}[{\cite{CFI}, \cite{Bijective_Pebble_Games}}]
\label{lem:wl-counting-logic-bijective-pebble-game}
Let \(\relstruc{A}\) and \(\relstruc{B}\) be two relational structures of arity at most \(k\), and \(\vec{a}\in V(\relstruc{A})^k\) and \(\vec{b}\in V(\relstruc{B})^k\)
two tuples of vertices. Then the following are equivalent:
{
\renewcommand{\labelenumi}{(\roman{enumi})}
\begin{enumerate}
\item Duplicator has a winning strategy in position \(\vec{a}\mapsto\vec{b}\) of the bijective \((k+1)\)-pebble game between \(\relstruc{A}\) and \(\relstruc{B}\),
\item for every \(\C^{k+1}\)-formula \(\phi(x_1,\dots,x_k)\), we have \((\relstruc{A},\vec{a})\models\phi\) if and only if \((\relstruc{B},\vec{b})\models\phi\),
\item the stable colors computed by \WL[k] for the tuples \(\vec{a}\) and \(\vec{b}\) agree.
\end{enumerate}
Further, every stable color class is definable by a single \(\C^{k+1}\)-formula.
}
\end{lemma}
In particular, the Weisfeiler-Leman dimension of a structure is precisely one less than the number of variables needed to define the structure in first-order counting logic.
\paragraph*{The CFI-construction.}\label{Section: CFI}
CFI-graphs are certain graphs with high Weisfeiler-Leman dimension \cite{CFI}.
To construct them, we start with a base graph \(G\), which is a connected simple graph, and a function \(f\colon E(G)\to\mathbb{F}_2\).
For a vertex \(v\in V(G)\), we denote the set of edges incident to \(v\) by \(E[v]\coloneqq\{uv\colon v\in N_G(v)\}\subseteq E(G)\). Now, to construct the CFI-graph \(\CFI(G,f)\), we replace each vertex \(v\in V(G)\)
by a gadget \(X_v\) which consists of inner vertices \(I_v\coloneqq \{v\}\times\{\vec{x}\in\mathbb{F}_2^{E[v]}\colon\sum\vec{x}=0\}\)
and outer vertices \(\{v\}\times\{(e,i)\colon e\in E[v], i\in\mathbb{F}_2\}\).
Inside each gadget, the inner and outer vertices each form an independent set, and an inner vertex \((v,\vec{x})\)
and outer vertex \((v,e,i)\) are connected by an edge if and only if \(\vec{x}_e=i\).
The resulting gadget for a vertex of degree \(3\) is depicted in Figure~\ref{Figure: CFI-gadget}.

\begin{figure}
\centering
\begin{tikzpicture}
\vertex[green]{x} {1,0};
\vertex[green]{x'}{2,0}; 

\vertex{k0}{0,1};
\vertex{k1}{1,1};
\vertex{k2}{2,1};
\vertex{k3}{3,1};

\vertex[red]{i} {2,2};
\vertex[red]{i'}{3,2};

\vertex[blue]{j} {0,2};
\vertex[blue]{j'}{1,2};

\draw (x) --  (k0) (x) --  (k1);
\draw (x') -- (k2) (x') -- (k3);

\draw (i) --  (k0) (i) --  (k2);
\draw (i') -- (k1) (i') -- (k3);

\draw (j) --  (k0) (j) --  (k3);
\draw (j') -- (k1) (j') -- (k2);
\end{tikzpicture}
\caption{A CFI-gadget for a vertex of degree \(3\), consisting of four inner vertices and three outer pairs}
\label{Figure: CFI-gadget}
\end{figure}
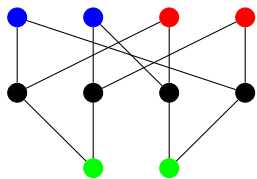

Next, we define the edge set between different gadgets. For every edge \(e=uv\in E(G)\),
we connect the outer vertices \((u,e,i)\) and \((v,e,j)\) if and only if \(i+j=f(e)\)
and add no further edges. Thus, corresponding outer vertex pairs \((u,e,\cdot)\) and \((v,e,\cdot)\)
are always connected by a matching, which is either \emph{untwisted} if \(f(e)=0\), or \emph{twisted} if \(f(e)=1\).

Finally, we define a vertex coloring on this graph. For every vertex \(v\), we turn the set~\(I_v\) of inner vertices
into a color class of size \(2^{d(v)-1}\). Moreover, we turn each outer pair \(\{(v,e,0),(v,e,1)\}\) into a color class
of size \(2\). This finishes the construction of CFI-graphs.

It turns out that for two functions \(f,g\colon E(G)\to\mathbb{F}_2\), we have \(\CFI(G,f)\simeqq\CFI(G,g)\) if and only if
\(\sum f=\sum g\), meaning that every even number of twists cancels out. Thus, we also write
\(\CFI(G,0)\) and \(\CFI(G,1)\) for the \emph{untwisted} and \emph{twisted} CFI-graphs over the base graph \(G\).

To understand the power of the Weisfeiler-Leman algorithm on CFI-graphs, it is convenient to study \emph{tree-width},
which is a graph parameter that intuitively measures how far a graph is from being a tree.
In this work, we do not need the formal definition of tree-width, and refer to~\cite{bodlaender_treewidth}.
The power of the Weisfeiler-Leman algorithm to distinguish CFI-graphs can now conveniently be expressed in terms of the tree-width of the base graphs, see \cite{WLdimvstw}.
\begin{lemma}\label{Lemma: CFI vs. tw}
For every base graph \(G\) of tree-width $\tw(G) \geq 2$, we have
\[\WLdim(\CFI(G,0))=\WLdim(\CFI(G,1))=\tw(G).\]
\end{lemma}
\begin{proof}
	We first show that CFI-graphs can be easily distinguished from other graphs.
	\begin{claim}\label{clm:wl-identifies-cfi}
		For  every base graph $G$,
		\WL[2] distinguishes CFI-graphs over $G$ from all other graphs.
	\end{claim}
	\begin{claimproof}
		Let $G$ be a base graph, $i \in \mathbb{F}_2$,
		and $H$ be some graph not isomorphic to $\CFI(G,0)$ or $\CFI(G,1)$.
		If $H$ has a color class of different size than the one of the same color in $\CFI(G,i)$,
		then \WL[2] certainly distinguishes $H$ and  $\CFI(G,i)$.
		Let $u \in V(G)$ be of degree $d$. Let $c$ be the color of inner vertices of the gadget of $u$ and $c_1,\dots,c_d$ be the colors of outer vertices of this gadget.
		Then every vertex of color $c$ has exactly one neighbor of each $c_i$ and no others.
		For each $i \in [d]$, the two vertices of color $c_i$ have the same number of neighbors in $c$, no common neighbor in $c$, and no neighbor in all $c_j$.
		Finally, for every two $c$-vertices~$u$ and~$v$,
		the number of $i\in[d]$, for which $u$ and $v$ have a different $c_i$-neighbor, is even.
		All these conditions can easily be recognized by \WL[2]
		and hence \WL[2] identifies all CFI-gadgets.
		If $e=uv \in E(G)$, let $c$ and $d$ be the color classes for the outer vertices of the gadgets for $u$ respectively $v$ for the edge $e$.
		Then $c$ and $d$ are of size two and connected by a matching,
		which is also easily identified by \WL[2].
		Hence, \WL[2] distinguishes $\CFI(G,i)$ from $H$.
	\end{claimproof}
	
	We present another modified CFI-construction.
	For a base graph $G$ and a function $f\colon E(G) \to \mathbb{F}_2$,
	let $\CFI'(G,f)$ be the graph obtained from $\CFI(G,f)$ in the following way:
	The vertex set of $\CFI'(G,f)$ is the set of all inner vertices of all CFI-gadgets in $\CFI(G,f)$.
	Two inner vertices $(u,\vec{x})$ and $(v,\vec{y})$ are adjacent in $\CFI'(G,f)$
	whenever $uv\in E(G)$ and $\vec{x}_{uv} + \vec{y}_{uv} = f(uv)$.
	This is equivalent to that there is a path of length $3$ between $(u,\vec{x})$ and $(v,\vec{y})$ 
	(namely the path $(u,\vec{x}), (u, uv, \vec{x}_{uv}), (v,uv, \vec{y}_{uv}), (v,\vec{y})$).
	These modified CFI-construction shares all important properties with the one presented here,
	in particular,  $\CFI'(G,f) \simeqq  \CFI'(G,g)$ if and only if $\sum f = \sum g$.
	The following claim is well-known (see~\cite{CFI,CFIvstw,WLdimvstw}):
	\begin{claim}\label{clm:cfi-inner-iff-inner-and-outer}
		For every $k\geq3$ and every base graph $G$,
		Spoiler wins the bijective $k$-pebble game on $\CFI'(G,0)$ and $\CFI'(G,1)$ if and only if 
		Spoiler wins the bijective $k$-pebble game on $\CFI(G,0)$ and $\CFI(G,1)$.
	\end{claim}
	The main insight to prove this lemma is that it is never beneficial for Spoiler to place a pebble on an outer vertex apart from the very end of the game because an outer vertex $(u,e,i)$ is uniquely identified by an inner vertex $(u, \vec{x})$ such that $\vec{x}_e = i$.
	
	\begin{claim}[\cite{WLdimvstw}]\label{clm:cfi-inner-tree-width}
		For every $k\geq 3$ and every base graph $G$,
		Spoiler wins the bijective $(k+1)$-pebble game on  $\CFI'(G,0)$ and $\CFI'(G,1)$
		if and only if $G$ has tree-width $\tw(G)\leq k$.
	\end{claim}
	
	Now finally let $G$ be a base graph of tree-width $\tw(G)\geq 2$.
	By Claims~\ref{clm:cfi-inner-iff-inner-and-outer} and~\ref{clm:cfi-inner-tree-width} and Lemma~\ref{lem:wl-counting-logic-bijective-pebble-game},
	\WL[k] distinguishes $\CFI(G,0)$ and $\CFI(G,1)$ if and only if $k \geq \tw(G)$.
	This implies $\WLdim(\CFI(G,i)) \geq \tw(G)$ for both $i \in \mathbb{F}_2$.
	By Claim~\ref{clm:wl-identifies-cfi},
	\WL[2] distinguishes CFI-graphs over $G$ from all other graphs.
	Since there are only two non-isomorphic CFI-graphs over $G$
	and \WL[\tw(G)] distinguishes them,
	\WL[\tw(G)] identifies CFI-graphs over $G$ and thus
	$\WLdim(\CFI(G,i)) \leq \tw(G)$ for both $i \in \mathbb{F}_2$.
	This implies $\WLdim(\CFI(G,0))=\WLdim(\CFI(G,1))=\tw(G)$.
\end{proof}

\section{Deciding identification for graphs with 5-bounded color classes}
\label{Section: deciding identification for graphs of color multiplicity 5}
As recently shown~\cite{WL2onCCS4}, identification
of a given graph with \(5\)-bounded color classes by \(\WL[2]\)
is polynomial-time decidable.
We extend this result to arbitrary dimensions of the Weisfeiler-Leman algorithm.
We adapt the approach of \cite{WL2onCCS4} and solve the separability for
\(2\)-induced \(k\)-ary coherent configurations with \(5\)-bounded fibers instead.
We generalize the elimination of interspaces containing a matching and interspaces of type \(2K_{1,2}\)
in order to reduce to \emph{star-free} \(2\)-induced \(k\)-ary coherent configurations.
By characterizing separability using certain automorphism groups, we provide a new reduction
of the separability problem for such configurations to the isomorphism problem
for graphs with bounded color classes, which can be solved in polynomial time.
\subsection{Functional basis relations and disjoint unions of stars}\label{sec:stars}

Let \(\cconf{C}\) be a \(k\)-ary coherent configuration and \(X,Y\in \F{\cconf{C}}\)
two distinct fibers.
A \emph{disjoint union of stars} between \(X\) and \(Y\)
is a basis relation \(S\in\cconf{C}|_2[Y,X]\) such that
every vertex in~\(Y\) is incident to exactly one outgoing edge in~\(S\).
If no interspace of \(\cconf{C}\) contains a disjoint union of stars,
\(\cconf{C}\) is called \emph{star-free}.

In this section, we want to show that whenever
\(\cconf{C}|_2[X,Y]\) contains a disjoint union of stars,
then \(\cconf{C}\) is separable if and only if
\(\cconf{C}\setminus X\coloneqq\cconf{C}[V(\cconf{C})\setminus X]\)
is separable. This allows us to reduce the separability problem to the class
of star-free \(k\)-ary coherent configurations.
As we will see, this is possible because the configuration \(\cconf{C}\)
is fully determined by the configuration \(\cconf{C}\setminus X\),
together with the information how the disjoint union of stars in \(\cconf{C}[X,Y]\)
attaches to \(Y\).

Slightly more generally, we call a binary basis relation \(S\in\cconf{C}|_2\)
\emph{functional} if every vertex of \(\cconf{C}\) is incident to at most
one outgoing edge in \(S\). Besides disjoint unions of stars, examples
of such basis relations are matchings or directed cycles within a fiber.

For a functional basis relation \(S\in\cconf{C}|_2\), we define a map
\begin{align*}
\nu_S\colon V(\cconf{C}) &\to V(\cconf{C}),\\
y&\mapsto
	\begin{cases}
		x &\text{if } yx\in S,\\
		y &\text{if no such } x \text{ exists}.
	\end{cases}
\end{align*}
Moreover, we denote for every \(I\subseteq[k]\) by \(\nu_S^I\colon V(\cconf{C})^k\to V(\cconf{C})^k\)
the map defined via \(\nu_S^I(y_1,\dots,y_k) \coloneqq (z_1,\dots,z_k)\),
where $z_i \coloneqq \nu_S(y_i)$ if $i\in I$ and $z_i \coloneqq y_i$ for every $i\in[k]\setminus I\).

\begin{lemma}\label{lem:stars:definable_maps_are_algebraic}
Let \(\cconf{C}\) be a \(k\)-ary coherent configuration,
and \(S\in\cconf{C}|_2\) a functional basis relation.
For every subset \(I\subseteq[k]\) and every basis relation \(R\in\cconf{C}\),
the set \(\nu_S^I(R)\) is also a basis relation of \(\cconf{C}\).

Moreover, for every algebraic isomorphism
\(f\colon\cconf{C}\to\cconf{D}\),
the basis relation \(f(S)\) is also functional in \(\cconf{D}\) 
and \(f\circ \nu_S^I=\nu_{f(S)}^I\circ f\) for every \(I\subseteq[k]\).
\begin{proof}
If \(I=\{i_1,\dots,i_\ell\}\), then we can write
\(\nu_S^I=\nu_S^{\{i_1\}}\circ\dots\circ\nu_S^{\{i_\ell\}}\).
As both claims of the lemma are compatible with compositions, it thus suffices to consider the case that \(I=\{i\}\) is a singleton.

Now, consider a colored variant \(\relstruc{C}\) of \(\cconf{C}\).
Because the partition \(\cconf{C}\) is stable under \WL[k]-refinement,
it is also equal to the partition into \(C^{k+1}\)-types by Lemma~\ref{lem:wl-counting-logic-bijective-pebble-game}.
In particular, for every formula \(\phi(x_1,\dots,x_\ell)\in C^{k+1}\) with \(\ell\leq k\),
the set \(\eval{\phi}_\relstruc{C}\coloneqq\{\vec{x}\in V(\cconf{C})^\ell\colon \relstruc{C},\vec{x}\models\phi\}\)
is a union of basis relations and conversely, every union of basis relations
is of the form \(\eval{\phi}_\relstruc{C}\) for an appropriate formula \(\phi\).

Thus, we find a formula \(\nu_S(x,y)\) which encodes the graph of the function \(\nu_S\),
and further, we find for every basis relation \(R\in\cconf{C}\), a formula \(\phi_R\) encoding \(R\).
But then, the formula
\(\psi_{\nu_S^i(R)}\coloneqq\exists x_{k+1} \nu_S(x_{k+1},x_i)\land\phi_R\frac{x_{k+1}}{x_i}\)
is a \(C^{k+1}\)-formula defining the set \(\nu_S^i(R)\). Thus, \(\nu_S^i(R)\) is again a union of basis relations.
If it was a proper union of basis relations, then there would exist some basis relation \(T\subsetneq \nu_S^i(R)\).
But then, the formula
\(\psi_R\land\exists x_{k+1} \nu_S(x_i,x_{k+1})\land \phi_T\frac{x_{k+1}}{x_i}\)
would define a proper subset of \(R\), which is impossible.
Thus, \(\nu_S^i(R)\) is a basis relation.

For the second claim, consider an algebraic isomorphism \(f\colon\cconf{C}\to\cconf{D}\)
and construct a colored variant of \(\cconf{D}\) in such a way that
\(f\) becomes a color-preserving map from \(\relstruc{C}\) to \(\relstruc{D}\).
Then \(f\) also preserves Weisfeiler-Leman colors and thus \(C^{k+1}\)-types.
In particular, this means that \(f(\eval{\phi}_\relstruc{C})=\eval{\phi}_{\relstruc{D}}\)
for every formula \(\phi\). In particular, for every basis relation \(R\in\cconf{C}\), we get
\[f(\nu_S^i(R))=f(\eval{\psi_{\nu_S^i(R)}}_\relstruc{C})=\eval{\psi_{\nu_S^i(R)}}_\relstruc{D}
=\nu_{f(S)}^i(f(R)),\]
which is what we wanted to show.
\end{proof}
\end{lemma}

From now on, we fix \(S\) to be a disjoint union of stars in \(\cconf{C}[X,Y]\).
In this case, \(\nu_S\) maps the fiber \(Y\) onto the fiber \(X\)
while fixing every other fiber.
Then, Lemma~\ref{lem:stars:definable_maps_are_algebraic} in particular implies that the basis relations of \(\cconf{C}\)
are fully determined by those of \(\cconf{C}\setminus X\).
Indeed, if \(R\) is a basis relation of \(\cconf{C}\)
and \(\vec{x}\in R\), let \(I\subseteq[k]\) be the set of \(X\)-components
of \(\vec{x}\). Then, for every \(i\in I\), we pick a \(\nu_S\)-preimage
\(y_i\) of \(x_i\). If we now define \(\vec{y}\coloneqq\vec{x}\left(\frac{y_i}{i}\right)_{i\in I}\),
we get \(\vec{x}=\nu_S^I(\vec{y})\) and thus~%
\(R=T^{\nu_S^I}\), where \(T\in\cconf{C}\setminus X\) is the basis relation containing \(\vec{y}\).

Together with Lemma \ref{lem:stars:definable_maps_are_algebraic},
this implies that the basis relations of \(\cconf{C}\) are precisely
those of the form \(T^{\nu_S^I}\) for \(T\in\cconf{C}\setminus X\) and \(I\subseteq[k]\).
Moreover, we can assume that \(I\) is a subset of the set of \(Y\)-components of \(T\).
These representations of basis relations of \(\cconf{C}\) are, however, not necessarily unique.

Before we are ready to eliminate all interspaces with a disjoint union of stars, we need one more technical observation
\begin{lemma}\label{Lemma: Definition of EqS}
If \(S\in \cconf{C}|_2[X,Y]\) is a disjoint union of stars, then the relation
\[\Eq{S}\coloneqq \{vw\in Y^2\colon \nu_S(v)=\nu_S(w)\}\]
is a union of \(2\)-ary basis relations.
Moreover, for every algebraic isomorphism~%
\(f\colon\cconf{C}\to \cconf{D}\),
we have \(f(\Eq{S})=\Eq{f(S)}\).
\begin{proof}
It suffices to assume that \(\cconf{C}\) is a \(2\)-ary coherent configuration.
Then, \(vw\in \Eq{S}\) if and only if for some \(z\) we have \(zv,zw\in S\).
If \(E\) is the basis relation of \(\cconf{C}\) containing \(vw\),
this is equivalent to
\(\intnum{E}{S,S^{-1}}\geq 1\),
where \(S^{-1}=\{yx\colon xy\in S\}\).
Because this property depends only on the basis relation containing \(vw\) and not on \(vw\) itself,
the first claim follows.

For the second claim, note that \(R\subseteq \Eq{S}\)
if and only if \(\nu_S^{\{1\}}(R)=S\).
Lemma~\ref{lem:stars:definable_maps_are_algebraic}
now implies
\[f(S)=(f\circ\nu_S^{\{1\}})(R)=(\nu_{f(S)}^{\{1\}}\circ f)(R)
=\nu_{f(S)}^{\{1\}}(f(R))\]
and thus \(f(R)\subseteq \Eq{f(S)}\).
\end{proof}
\end{lemma}

\noindent We are now ready to prove one of the implications of the claimed equivalence:
\begin{lemma}\label{Lemma: Elimination of disjoint unions of stars I}
Let \(\cconf{C}\) be a \(k\)-ary coherent configuration,
\(X,Y\in \F{\cconf{C}}\) two distinct fibers
and \(S\in\cconf{C}|_2[X,Y]\) a disjoint union of stars between \(X\) and \(Y\).
If \(\cconf{C}\setminus X\) is separable, then so is \(\cconf{C}\).
\begin{proof}
Assume that \(\cconf{C}\setminus X\) is separable
and let \(f\colon\cconf{C}\to\cconf{D}\)
be an algebraic isomorphism.
Then \(f\) restricts to an algebraic isomorphism~\(f_0\colon\cconf{C}\setminus X\to\cconf{D}\setminus f(X)\),
which is induced by a combinatorial isomorphism~%
\(\phi_0\colon\cconf{C}\setminus X\to\cconf{D}\setminus f(X)\)
because \(\cconf{C}\setminus X\) is separable.
We extend this map to all of \(\cconf{C}\) by defining
\begin{align*}
\phi\colon V(\cconf{C})&\to V(\cconf{D}),\\
x&\mapsto\begin{cases}
\phi_0(x)&\text{if } x\notin X,\\
(\nu_{f(S)}\circ\phi_0)(y)&\text{if } x\in X \text{ and } x=\nu_S(y) \text{ for some } y\in Y.
\end{cases}
\end{align*}
We claim that \(\phi\) is a well-defined combinatorial isomorphism
inducing \(f\).

To see well-definedness, we need to argue that the definition of \(\phi(x)\) for vertices \(x\in X\)
does not depend on the choice of its \(\nu_S\)-preimage \(y\in Y\).
But indeed, if~\(\nu_S(y)=\nu_S(y')\), then \(yy'\in \Eq{S}\).
By Lemma~\ref{Lemma: Definition of EqS}, it holds that \(f(\Eq{S})=\Eq{f(S)}\).
As \(\phi_0\) induces \(f_0\), we get \(\phi_0(y)\phi_0(y')\in \Eq{f(S)}\).
But this means that \(\nu_{f(S)}(\phi_0(y))=\nu_{f(S)}(\phi_0(y'))\) which implies that~\(\phi\) is well-defined.
Moreover, it is a bijection as \(\phi_0\) is a bijection
and \(\phi|_X\colon X\to f(X)\) is a bijection.
To show that \(\phi\) induces \(f\), consider a basis relation \(\nu_S^I(R)\in\cconf{C}\)
where \(R\in\cconf{C}\setminus X\) and \(I\) is a subset of the \(Y\)-components of \(X\).
Applying \(\phi\) to this basis relation yields
\begin{align*}
\phi(\nu_S^I(R))
&=\left(\phi\circ\nu_S^I\right)(R).\\
	\intertext{Using the definition of \(\phi\), this is equal to}
&=\left(\nu_{f(S)}^I\circ \phi_0\right)(R)
=\left(\nu_{f(S)}^I\circ f_0\right)(R)
=\left(\nu_{f(S)}^I\circ f\right)(R),\\
	\intertext{which, by Lemma~\ref{Lemma: Definition of EqS}, is equal to}
&=\left(f\circ \nu_{S}^I\right)(R)
=f \left(\nu_S^I(R)\right).
\end{align*}
This proves that \(\phi\) induces \(f\) and thus in particular that \(\phi\) is a combinatorial isomorphism.
Thus, \(\cconf{C}\) is separable.
\end{proof}
\end{lemma}

Next, we want to prove the converse of Lemma~\ref{Lemma: Elimination of disjoint unions of stars I}, which will allow us to
eliminate all disjoint unions of stars without affecting separability of the configuration.
For this, we need one auxiliary lemma:
% Usage: \attachStar{configuration}{disj. union of cliques}
\newcommand{\attachStar}[2]{#1_{#2}^\star}
\begin{lemma}\label{Lemma: Construction of C_X}
For every \(k\)-ary coherent configuration \(\cconf{C}\),
fiber \(Y\in\F{\cconf{C}}\) and every union of basis relations \(E\in\left(\cconf{C}|_2[Y]\right)^\cup\)
which forms an equivalence relation on \(Y\),
there exists an extension \(\attachStar{\cconf{C}}{E}\) of \(\cconf{C}\) by a single fiber \(X\)
such that the interspace \(\attachStar{\cconf{C}}{E}[X,Y]\) contains a disjoint union
of stars \(S\) satisfying \(E=\Eq{S}\).
Moreover, this extension is unique up to combinatorial isomorphism.

Furthermore, the construction satisfies the following properties:
\begin{enumerate}
\item\label{Property: Construction of C_X, compatibility with 2-skeletons}
	For every \(\cconf{C}\) and \(E\) as above, we have
	\(\attachStar{(\cconf{C}|_2)}{E}=\left.\left(\attachStar{\cconf{C}}{E}\right)\right|_2\).
\item\label{Property: Construction of C_X, inverse to removing X}
	For every configuration \(\cconf{C}\) and every interspace \(\cconf{C}[X,Y]\) that contains a disjoint union of stars \(S\),
	we have	\(\cconf{C}\simeqq\attachStar{(\cconf{C}\setminus X)}{\Eq{S}}\).
\end{enumerate}
\begin{proof}
Define \(X\coloneqq Y/E\) as the quotient of \(Y\) by the equivalence relation \(E\),
and \(\nu_S\colon Y\to X\) as the natural projection.
Now, we define \(V(\attachStar{\cconf{C}}{E})\coloneqq V(\cconf{C})\dotcup X\),
and extend \(\nu_{S}\) to all of \(V(\attachStar{\cconf{C}}{E})\) by letting it act
as the identity on every fiber besides \(Y\). Once again, we write \(\nu_S^I\)
for the map \(V(\attachStar{\cconf{C}}{E})^k\to V(\attachStar{\cconf{C}}{E})^k\) defined by \(\nu_S^I(y_1,\dots,y_k)\coloneqq (z_1,\dots,z_k)\),
where \(z_i\coloneqq \nu_S^I(y_i)\) for all \(i\in I\) and \(z_i\coloneqq y_i\) for all \(i\notin I\).
To simplify notation, we also abbreviate this map as \((I)\).

By Lemma~\ref{lem:stars:definable_maps_are_algebraic},
all of the sets \(R^{(I)}\) with \(R\in\mathcal{C}\) and \(I\) a subset of the \(Y\)-components of \(R\) must
be basis relations of \(\attachStar{\cconf{C}}{E}\) if we want to make \(\attachStar{\cconf{C}}{E}\) coherent.
Hence, we can only define the basis relations of \(\attachStar{\cconf{C}}{E}\) in a unique way:
\[\attachStar{\cconf{C}}{E}\coloneqq\left\{R^{(I)}\colon R\in\cconf{C},I\subseteq[k], R_I\subseteq Y^{|I|}\right\}.\]

\begin{claim}
\(\attachStar{\cconf{C}}{E}\) is a \(k\)-ary rainbow with \(\attachStar{\cconf{C}}{E}\setminus X=\cconf{C}\).
\begin{claimproof}
The fact that every \(k\)-tuple in \(\attachStar{\cconf{C}}{E}\) is contained
in at least one of these basis relations follows from surjectivity of \(\nu_S\).
To see that every two distinct basis relations are disjoint,
we note that for every basis relation \(R^{(I)}\in\attachStar{\cconf{C}}{E}\), the set~\(I\)
can be reconstructed as the set of \(X\)-components of \(R^{(I)}\).
This implies that if \(R^{(I)}\cap {R'}^{(I')}\neq\emptyset\),
then \(I=I'\). For every $k$-tuple~\(\vec{x}\) in this intersection,
we can write \(\vec{x}=\vec{y}^{(I)}=(\vec{y}')^{(I)}\) for some \(\vec{y}\in R\)
and \(\vec{y}\in R'\). Then for every \(i\in I\), we have
\(\nu_{S}(y_i)=x_i=\nu_{S}(y_i')\) and thus \(y_iy_i'\in E\).

Now, fix some \(i\in I\), and let \(R_i\) be the basis relation
containing \(\vec{y}\frac{y_i'}{i}\).
For all \(j\neq i\), we set
\[R_j\coloneqq\{\vec{z}\in V(\cconf{C})\colon z_iz_j\in E\}.\]
The sets \(R_j\) are unions of basis relations, and we have
\[\intnum{R}{R_1,\dots,R_k}\geq 1,\]
because \(\vec{y}\frac{y_i'}{j}\in R_j\) for all \(j\in[k]\).
But as the relations \(R_j\) for \(j\neq i\) also ensure that any witness~\(z\)
to this inequality must satisfy \(y_iz\in E\),
we get that for every~\(\vec{y}\in R\), there is some~%
\(z\in V(\cconf{C})\) with \(y_iz\in E\) such that \(\vec{y}\frac{z}{i}\in R_i\).
We can iterate this construction with \(R_i\) instead of \(R\)
and a different \(i\in I\). If we do this for all \(i\in I\),
we finally get that for every \(\vec{y}\in R\) and every~%
\(i\in I\), there is some~\(z_i\in V(\cconf{C})\)
such that \(y_iz_i\in E\) and furthermore,
these~\(z_i\) satisfy \(\vec{y}\left(\frac{z_i}{i}\right)_{i\in I}\in R'\).
But this means~\(R^{(I)}\subseteq {R'}^{(I)}\) and
by symmetry of the argument,~\(R^{(I)}={R'}^{(I)}\).

This concludes the proof that the basis relations \(R^{(I)}\) form a partition
of \(V(\attachStar{\cconf{C}}{E})^k\).
The fact that this partition is compatible with equality types follows from the fact
that for~\(\vec{x}=\vec{y}^{(I)}\in R^{(I)}\) we have
\(x_i=x_j\) if and only if
either~\(i,j\in I\) and~\(y_iy_j\in E\)
or~\(i,j\notin I\) and~\(y_i=y_j\).
The fact that the partition is compatible with permutations follows from the observation that
\(\bigl(R^{(I)}\bigr)^\sigma=(R^\sigma)^{(I^\sigma)}\)
 for every permutation \(\sigma\) on \([k]\).
This concludes the proof that \(\attachStar{\cconf{C}}{E}\) is a \(k\)-ary rainbow.
As furthermore \(R^{(\emptyset)}=R\) for all \(R\in\cconf{C}\),
this rainbow extends \(\cconf{C}\).
\end{claimproof}
\end{claim}

In order to finish the proof that \(\attachStar{\cconf{C}}{E}\) is a \(k\)-ary coherent configuration,
it remains to show that the intersection numbers of \(\attachStar{\cconf{C}}{E}\) are well-defined,
that is, independent on the specific choice of vertex tuple in the basis relation.
For this, we write \(R\sim_I R'\) for basis relations \(R,R\in\cconf{C}\)
if \(R^{(I)}={R'}^{(I)}\).
Moreover, we write
\[R/I\coloneqq\bigunion_{\substack{R'\in\cconf{C}\\ R'\sim_I R}} R'.\]

Now, to prove coherence of \(\attachStar{\cconf{C}}{E}\),
let \(R^{(I)},T_1^{(I_1)},\dots,T_k^{(I_k)}\in\attachStar{\cconf{C}}{E}\)
be arbitrary basis relations and pick some representative
\(\vec{x}=\vec{y}^{(I)}\in R^{(I)}\)
with \(\vec{y}\in R\).
For every \(z\in V(\cconf{C})\), we then have
\[\vec{x}\frac{z}{i}
=\vec{y}^{(I)}\frac{z}{i}
=\left(\vec{y}\frac{z}{i}\right)^{(I\setminus\{i\})}.\]
This means that \(\vec{x}\frac{z}{i}\in T_i^{(I_i)}\)
if and only if \(I_i=I\setminus\{i\}\) and
\(\vec{y}\frac{z}{i}\in T_i[\sim_{I_i}]\). 
For all \(x\in X\) and \(y\in Y\) such that \(x=\nu_S(y)\),
the analogous calculation instead yields
\[\vec{x}\frac{x}{i}
=\vec{y}^{(I)}\frac{\nu_{S}(y)}{i}
=\left(\vec{y}\frac{y}{i}\right)^{(I\cup\{i\})}.\]
This means that \(\vec{x}\frac{x}{i}\in T_i^{(I_i)}\)
if and only if
\(I_i=I\cup\{i\}\) and
\(\vec{y}\frac{y}{i}\in T_i[\sim_{I_i}]\).
As one of the~\(S\)-neighbors of \(x\) satisfies this last condition if and only if
every~\(S\)-neighbor does, we overcount the number of such \(x\in X\) by a factor of
\(\extnum{X}{S}\) if we count the number of such \(y\in Y\) instead.

In total, this means that if for some \(\vec{x}\in R^{(I)}\) we define
\[\intnum{\vec{x}}{T_1^{(I_1)},\dots,T_k^{(I_k)}}\coloneqq
\left|\left\{y\in V(\cconf{C})\colon \vec{x}\frac{y}{i}\in R_i \text{ for all } i\in[k]\right\}\right|,\]
then this intersection number can be expressed as
\begin{equation}
\label{Equation: Intersection numbers of D_X in proof of elimination of stars}
\displaystyle \intnum{\vec{x}}{T_1^{(I_1)},\dots,T_k^{(I_k)}}
=\begin{cases}
\intnum{R}{T_1[\sim_{I_1}],\dots,T_k[\sim_{I_k}]}
	&\text{if } I_i=I\setminus\{i\} \text{ for all } i\in[k],\vspace{2mm}\\
\displaystyle\frac{\intnum{R}{T_1[\sim_{I_1}],\dots,T_k[\sim_{I_k}]}}{\extnum{X}{S}}
	&\text{if } I_i=I\cup\{i\} \text{ for all } i\in[k],\vspace{2mm}\\
0	&\text{otherwise.}
\end{cases}
\end{equation}
As the right-hand side of Equation~\ref{Equation: Intersection numbers of D_X in proof of elimination of stars}
does not depend on \(\vec{x}\) but only on the basis relation \(R\) containing \(\vec{y}\)
which can be chosen to be the same for all \(\vec{x}=\vec{y}^{(I)}\in R^{(I)}\),
the intersection number \(\intnum{R^{(I)}}{T_1^{(I_1)},\dots,T_k^{(I_k)}}\) is well-defined.
This proves coherence of \(\attachStar{\cconf{C}}{E}\).

Property~\ref{Property: Construction of C_X, compatibility with 2-skeletons} follows from the observation that
for all \(I\subseteq[k]\), the \([2]\)-face of \(R^{(I)}\)
is given by \((R_{[2]})^{(I\cap [2])}\).
Property~\ref{Property: Construction of C_X, inverse to removing X} follows from uniqueness of the extension.
\end{proof}
\end{lemma}

\noindent Now, we are ready to prove the converse of Lemma~\ref{Lemma: Elimination of disjoint unions of stars I}.
\begin{lemma}\label{Lemma: Elimination of disjoint unions of stars II}
Let \(\cconf{C}\) be a \(k\)-ary coherent configuration,
\(X,Y\in \F{\cconf{C}}\) two distinct fibers,
and \(S\in\cconf{C}|_2[X,Y]\) a disjoint union of stars between \(X\) and \(Y\).
\begin{enumerate}
\item If \(\cconf{C}\) is separable, so is \(\cconf{C}\setminus X\).
\item If \(\cconf{C}\) is \(2\)-induced, so is \(\cconf{C}\setminus X\).
\end{enumerate}
\begin{proof}
We start with the first claim.
Assume that \(\cconf{C}\) is separable and let
\(f_0\colon\cconf{C}\setminus X\to\cconf{D}\) be an arbitrary algebraic isomorphism.
Then \(E\coloneqq f_0(\Eq{S})\) is a union of \(2\)-ary basis relations
in \(\cconf{D}[f_0(Y)]\) which forms an equivalence relation.
We want to construct an extension~\(f\colon\cconf{C}\to \attachStar{\cconf{D}}{E}\) of \(f\).
Then, by separability of \(\cconf{C}\), we would find a combinatorial
isomorphism~\(\phi\) that induces~\(f\)
and whose restriction to \(\cconf{C}\setminus X\)
induces \(f_0\).

For this, recall that the basis relations of \(\cconf{C}\) are all of the form
\(R^{(I)}=\nu_S^I(R)\) with \(R\in\cconf{C}\setminus X\) and \(I\) a subset of the \(Y\)-components of \(R\);
and the basis relations of \(\attachStar{\cconf{D}}{E}\) are of the form \(T^{(I)}=\nu_{f(S)}^I(T)\)
for \(T\in\cconf{D}\) and \(I\) a subset of the \(f(Y)\)-components of \(T\).
We define the extension \(f\) of \(f_0\) by
\[f\left(R^{(I)}\right)\coloneqq \left(f_0(R)\right)^{(I)}.\]
\begin{claim}
The map \(f\) is a well-defined extension of \(f_0\).
\begin{claimproof}
To see that this map is well-defined, note that the set \(I\) is precisely the set of \(X\)\nobreakdash-components
of \(R^{(I)}\) and is thus determined by the basis relation.
Thus, we only need to show that whenever \(R^{(I)}={R'}^{(I)}\), then also
\((f_0(R))^{(I)}=\left(f_0(R')\right)^{(I)}\).
The former is the case
if and only if \(R\) and \(R'\) can be transformed into one another
by replacing for all \(\vec{x}\in R\) the \(X\)-components of \(\vec{x}\) by
one or multiple \(\Eq{S}\)-equivalent components. As the analogous claim is true in \(\attachStar{\cconf{D}}{E}\)
and as \(f_0\) sends \(\Eq{S}\) to \(\Eq{f(S)}\) by construction, it follows that
\(f\) is well-defined.
Moreover, \(f\) extends \(f_0\), as for \(I=\emptyset\),
we get \(f(R)=f_0(R)\) for all \(R\in\cconf{C}\setminus X\).
\end{claimproof}
\end{claim}

\begin{claim}
The map \(f\) is an algebraic isomorphism.
\begin{claimproof}
The intersection numbers of both
\(\cconf{C}\simeqq\attachStar{(\cconf{C}\setminus X)}{\Eq{S}}\) and \(\attachStar{\cconf{D}}{E}\) are
determined in terms of those of \(\cconf{C}\setminus X\) and \(\cconf{D}\)
by Equation~\eqref{Equation: Intersection numbers of D_X in proof of elimination of stars}.
To see that \(f\) preserves the right-hand side of this formula,
recall that we defined
\[R/I=\bigcup_{\substack{R'\in\cconf{C}\setminus X\\ R^{(I)}=(R')^{(I)}}} R'.\]
It now remains to show
that \(f_0(R/I)=f_0(R)/I\) for all \(R\in\cconf{C}\setminus X\)
and every subset \(I\) of the \(Y\)-components of \(R\).

We again use the observation that \(R^{(I)}=(R')^{(I)}\) if and only if \(R\) can be obtained from \(R'\)
by replacing some of its \(Y\)-components by \(\Eq{S}\)-neighbors, and that the analogous statement is true in \(\attachStar{\cconf{D}}{E}\) with respect to \(\Eq{f(S)}\). Therefore, we have \(R^{(I)}=(R')^{(I)}\) if and only if
\(\left(f_0(R)\right)^{(I)}=\left(f_0(R')\right)^{(I)}\). This implies \(f(R/I)=f(R)/I\) and thus that
\(f\) preserves intersection numbers.
\end{claimproof}
\end{claim}

We have now successfully extended the algebraic isomorphism \(f_0\colon\cconf{C}\setminus X\to\cconf{D}\)
to an algebraic isomorphism \(f\colon \cconf{C}\to\attachStar{\cconf{D}}{E}\).
Because \(\cconf{C}\) is separable, \(f\) is induced by a combinatorial isomorphism \(\phi\)
which restricts to an isomorphism \(\phi_0\) inducing \(f_0\).

%%%%%%%%%%%%%% 2-inducedness %%%%%%%%%%%%%%%%%%%%%%%%%%%%%%%%%%%%%%%%%%%
For the second claim, assume that \(\cconf{C}\) is \(2\)-induced.
To show that \(\cconf{C}\setminus X=\CC{k}((\cconf{C}\setminus X)|_2)\),
we first note that the left-hand side is a \(k\)-ary coherent configuration
on \(V(\cconf{C})\setminus X\) that is at least as fine as \((\cconf{C}\setminus X)|_2\)
and thus also at least as fine as \(\CC{k}((\cconf{C}\setminus X)|_2)\).
For the converse, consider the \(k\)-ary coherent configuration
\(\cconf{C}'\coloneqq\attachStar{\CC{k}((\cconf{C}\setminus X)|_2)}{\Eq{S}}\).
By choosing the obvious bijection between the fiber \(X\in\F{\cconf{C}}\) and the
additional fiber of \(\cconf{C}'\), we can assume that \(\cconf{C}'\) is a \(k\)-ary coherent configuration on \(V(\cconf{C})\), which is furthermore at most as fine as \(\cconf{C}\).
Now, note that
\[\mathcal{C'}|_2
=\attachStar{\Bigl(\CC{k}\bigl((\cconf{C}\setminus X)|_2\bigr)\big|_2\Bigr)}{\Eq{S}}
=\attachStar{\bigl((\cconf{C}\setminus X)|_2\bigr)}{\Eq{S}}
=\attachStar{(\cconf{C}\setminus X)}{\Eq{S}}\big|_2
=\cconf{C}|_2,\]
where the first and third equality follow from Property \ref{Property: Construction of C_X, compatibility with 2-skeletons} of Lemma~\ref{Lemma: Construction of C_X}, the second equality follows from the fact that \(\cconf{C}\setminus X\) is a \(k\)-ary coherent configuration,
and the last equality follows from Property~\ref{Property: Construction of C_X, inverse to removing X}
of Lemma~\ref{Lemma: Construction of C_X}.
As \(\cconf{C}\) is \(2\)-induced, this yields \(\cconf{C}'\feq\CC{k}(\cconf{C}|_2)=\cconf{C}\)
and thus \(\cconf{C}=\cconf{C}'\).
Finally,
\[\cconf{C}\setminus X
 =\cconf{C}'\setminus X
 =\attachStar{(\CC{k}((\cconf{C}\setminus X)|_2))}{E}\setminus X
=\CC{k}((\cconf{C}\setminus X)|_2)\]
shows that \(\cconf{C}\setminus X\) is \(2\)-induced.
\end{proof}
\end{lemma}

Lemma~\ref{Lemma: Elimination of disjoint unions of stars I} together with
Lemma~\ref{Lemma: Elimination of disjoint unions of stars II} allows us to remove fibers from coherent configurations containing
a disjoint union of stars without affecting the separability of the configuration.
Thus, we have reduced the separability problem
for arbitrary \(k\)-ary coherent configurations to the same problem for star-free configurations.
Furthermore by Lemma~\ref{Lemma: Elimination of disjoint unions of stars II}, this process preserves \(2\)-inducedness of a configuration. Thus, in order to solve the separability problem
for \(2\)-induced \(k\)-ary coherent configurations, it only remains to consider \(2\)-induced star-free configurations.
This simultaneously generalizes the elimination of interspaces containing a matching
and the elimination of fibers of size \(2\) from~\cite{WL2onCCS4}.

\subsection{Structure of \texorpdfstring{\(k\)}{k}-ary coherent configurations with 5-bounded fibers}

\begin{figure}
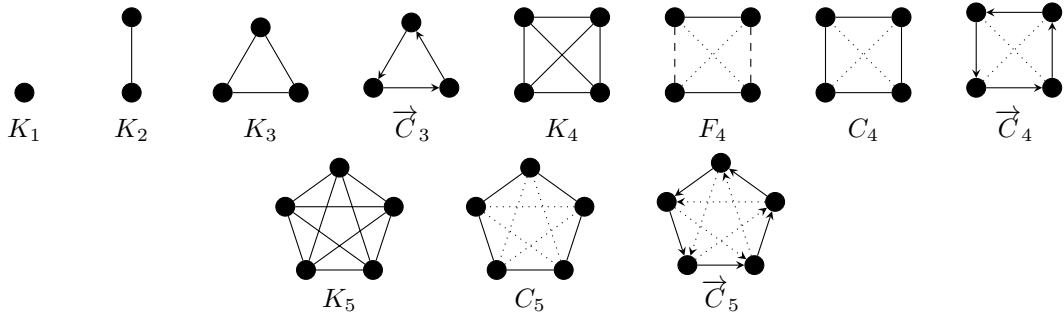

\newcommand*\KOne   {\vertex{A}{0,0};}
\newcommand*\KTwo   {\vertex{A}{0,0};\vertex{B}{0,1};                                \draw (A) -- (B);}
\newcommand*\KThree {\vertex{A}{0,0};\vertex{B}{1,0};\vertex{C}{0.5,1.732/2};        \draw (A) -- (B) -- (C) -- (A);}
\newcommand*\dCThree{\vertex{A}{0,0};\vertex{B}{1,0};\vertex{C}{0.5,1.732/2};        \draw[->] (A) edge (B) (B) edge (C) (C) edge (A);}
\newcommand*\KFour  {\vertex{A}{0,0};\vertex{B}{1,0};\vertex{C}{1,1};\vertex{D}{0,1};\draw (A)--(B)--(C)--(D)--(A)--(C) (B)--(D);}
\newcommand*\FFour  {\vertex{A}{0,0};\vertex{B}{1,0};\vertex{C}{1,1};\vertex{D}{0,1};\draw (A)--(B) (C)--(D);\draw[dotted](A)--(C) (B)--(D);\draw[dashed](A)--(D) (B)--(C);}
\newcommand*\CFour  {\vertex{A}{0,0};\vertex{B}{1,0};\vertex{C}{1,1};\vertex{D}{0,1};\draw (A)--(B)--(C)--(D)--(A);\draw[dotted] (A)--(C) (B)--(D);}
\newcommand*\dCFour {\vertex{A}{0,0};\vertex{B}{1,0};\vertex{C}{1,1};\vertex{D}{0,1};\draw[->] (A)edge(B)(B)edge(C)(C)edge(D)(D)edge(A);\draw[dotted] (A)--(C) (B)--(D);}
\newcommand*\KFive  {\vertex{A}{18:0.75};\vertex{B}{90:0.75};\vertex{C}{162:0.75};\vertex{D}{234:0.75};\vertex{E}{306:0.75};\draw (A)--(B)--(C)--(D)--(E)--(A)--(C)--(E)--(B)--(D)--(A);}
\newcommand*\CFive  {\vertex{A}{18:0.75};\vertex{B}{90:0.75};\vertex{C}{162:0.75};\vertex{D}{234:0.75};\vertex{E}{306:0.75};\draw (A)--(B)--(C)--(D)--(E)--(A);\draw[dotted] (A)--(C)--(E)--(B)--(D)--(A);}
\newcommand*\dCFive  {\vertex{A}{18:0.75};\vertex{B}{90:0.75};\vertex{C}{162:0.75};\vertex{D}{234:0.75};\vertex{E}{306:0.75};\draw[->] (A)edge(B)(B)edge(C)(C)edge(D)(D)edge(E)(E)edge(A);\draw[dotted, ->] (A)edge(C)(C)edge(E)(E)edge(B)(B)edge(D)(D)edge(A);}

\newcommand{\name}[2][0]{\node at (#1,-0.5) {#2};}

\centering
\tikz{\KOne    \name{\(K_1\)}}
\hfill
\tikz{\KTwo    \name{\(K_2\)}}
\hfill
\tikz{\KThree  \name[0.5]{\(K_3\)}}
\hfill
\tikz{\dCThree \name[0.5]{\(\overrightarrow{C}_3\)}}
\hfill
\tikz{\KFour   \name[0.5]{\(K_4\)}}
\hfill
\tikz{\FFour   \name[0.5]{\(F_4\)}}
\hfill
\tikz{\CFour   \name[0.5]{\(C_4\)}}
\hfill
\tikz{\dCFour  \name[0.5]{\(\overrightarrow{C}_4\)}}\\
\tikz{\KFive  \node at (0,-1) {\(K_5\)};}
\qquad
\tikz{\CFive  \node at (0,-1) {\(C_5\)};}
\qquad
\tikz{\dCFive  \node at (0,-1) {\(\overrightarrow{C}_5\)};}

\caption{The complete list of \(2\)-ary coherent configurations on a single fiber of order up to \(4\) from \cite{WL2onCCS4},
and the three \(2\)-ary coherent configurations on a single fiber of order \(5\) from \cite{FibersOrder5}.}
\label{Figure: two-dimensional association schemes of order up to 5}
\end{figure}

A list of all isomorphism types of \(2\)-ary coherent configurations on a single fiber of order at most \(5\) is known~\cite{WL2onCCS4, FibersOrder5},
see Figure~\ref{Figure: two-dimensional association schemes of order up to 5}.
Furthermore, the authors of \cite{WL2onCCS4} also give a complete list of possible interspaces
between fibers of order up to \(4\).
Here, it is always possible that the interspace \(\cconf{C}[X,Y]\) is \emph{uniform},
meaning that it consists of only a single basis relation \(X\times Y\).
Uniform interspaces are usually easy to handle and the other interspace types are the more interesting ones.
As we already eliminated interspaces containing a disjoint union of stars,
Figure~\ref{Figure: 2-interspaces between two fibers}
contains only the two remaining non-uniform and star-free interspaces.

In \cite[Lemma 5.2]{WL2onCCS4}, it was shown that both of these interspaces
enforce the existence of a matching in both incident fibers. These matchings
are denoted by dotted lines in Figure~\ref{Figure: 2-interspaces between two fibers}.

As there can be no non-uniform interspaces between fibers of coprime orders \cite[Lemma 3.1]{WL2onCCS4},
we are thus only missing an enumeration of the non-uniform star-free interspaces between two fibers
of order \(5\) each. However, as every such interspace must contain a relation which has degree \(2\) in both fibers,
the only candidate for such an interspace is the \(C_{10}\), the interspace containing a basis relation forming a \(10\)-cycle
between the two fibers. But this \(10\)-cycle induces a matching between opposite vertices.
Thus, in a star-free \(2\)-ary coherent configuration with \(5\)-bounded fibers, there can be no non-uniform interspaces
incident to any fibers of order \(5\).

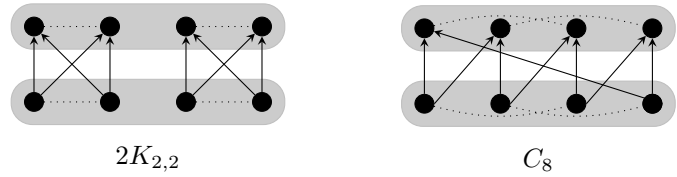
\begin{figure}
% Usage: \fiber[rotation]{name}{order}{center}
\newcommand{\fiber}[4][0]{
	\draw[rounded corners=8pt, fill=gray!40!white, draw=gray!50!white, rotate around={#1:(#4)}]
		($(#4)+(-#3/2+0.2,-0.3)$) rectangle ($(#4)+(#3/2-0.2,+0.3)$) {};
	\foreach \i in {1,...,#3}{
		\vertex{#2\i}{$(#4)+{\i-#3/2-0.5}*({cos(#1)}, {sin(#1)})$};
	}
}
\centering
\begin{tikzpicture}
	\fiber{A}{4}{0,0}
	\fiber{B}{4}{0,1}
	\foreach \i\j in {1/1,1/2,2/1, 2/2, 3/3, 3/4, 4/3, 4/4}{
		\draw[->] (A\i) -- (B\j);
	}
	\draw[dotted] (A1) -- (A2) (A3) -- (A4);
	\draw[dotted] (B1) -- (B2) (B3) -- (B4);
	
	\node at (0,-0.75) {\(2K_{2,2}\)};
\end{tikzpicture}
\qquad\qquad
\begin{tikzpicture}
	\fiber{A}{4}{0,0}
	\fiber{B}{4}{0,1}
	\vertex{A0}{A4};
	\foreach \i[evaluate=\i as \imm using (\i-1)] in {1,...,4}{	
		\draw[->] (A\i) -- (B\i);
		\draw[->] (A\imm) -- (B\i);
	}
	
	\draw[dotted] (B1) to[bend left=15] (B3) (B2) to[bend left=15] (B4);
	\draw[dotted] (A1) to[bend right=15] (A3) (A2) to[bend right=15] (A4);
	
	\node at (0,-0.75) {\(C_8\)};
\end{tikzpicture}
\caption{All non-uniform and star-free interspace types between two fibers of order up to \(5\).
	In each case, one of the basis relations between the two fibers is missing
	and can be reconstructed as the complement of the one drawn.}
\label{Figure: 2-interspaces between two fibers}
\end{figure}

In the \(k\)-dimensional case, interspaces are generally not between two, but between \(k\)
fibers and their enumeration and analysis is more complicated.
But as we are only dealing with \(2\)-induced \(k\)-ary coherent configuration,
the two-dimensional interspaces suffice for our purposes. 

We need the following fact about star-free (and indeed even matching-free) configurations:
\begin{lemma}[{\cite[Lemma 6.2]{WL2onCCS4}}]\label{Lemma: C8-interspaces not incident}
In a star-free \(2\)-ary coherent configuration, no two distinct interspaces of type \(C_8\)
can be incident to a common fiber.
\end{lemma}

This allows us to prove the following generalization of \cite[Lemma 7.1]{WL2onCCS4}:
\begin{lemma}\label{Lemma: Every algebraic isomorphism is induced by combinatorial isomorphism on every fiber}
Let \(\cconf{C}\) be a star-free, \(2\)-ary coherent configuration with \(5\)-bounded fibers.
If \(f\colon\cconf{C}\to\mathcal{D}\) is an algebraic isomorphism, then
there exists a combinatorial isomorphism \(\phi\colon V(\cconf{C})\to V(\mathcal{D})\)
which satisfies \(\phi(X)=f(X)\) for every fiber \(X\in\F{\cconf{C}}\).
\begin{proof}
As every \(2\)-ary coherent configuration with at most \(8\) vertices is separable \cite{WL2onCCS4},
we can pick for every \(C_8\)-interspace \(\cconf{C}[X,Y]\)
a combinatorial isomorphism \(\phi_{X,Y}\colon X\cup Y\to f(X)\cup f(Y)\)
which induces \(f|_{\cconf{C}[X\cup Y]}\) and for every fiber \(Z\) not incident to a \(C_8\)-interspace
an isomorphism \(\phi_Z\colon Z\to f(Z)\) inducing \(f|_{\cconf{C}[Z]}\).

By Lemma~\ref{Lemma: C8-interspaces not incident},
the domains of these partial combinatorial isomorphisms form a partition
of \(V(\cconf{C})\). Thus, we can combine these bijections to get a total bijection
\(\phi\colon V(\cconf{C})\to V(\mathcal{D})\). It remains to show that \(\phi\)
is a combinatorial isomorphism on every interspace. 

On uniform interspaces, the claim is immediate and further it is true by construction on interspaces of type \(C_8\).
The only remaining interspace type is thus \(2K_{2,2}\).
But here, the interspace is uniquely determined by its isomorphism-type
and the matchings it induces in the incident fibers. On such interspaces \(\phi\) thus
either induces \(f\) or switches the two binary basis relations.
In both cases, \(\phi\) is a combinatorial isomorphism on \(\cconf{C}\).
\end{proof}
\end{lemma}

We call an algebraic automorphism \(f\) of a \(k\)-ary coherent configuration \(\cconf{C}\) \emph{strict}
if it fixes every fiber,
i.e. satisfies \(f(X)=X\) for every fiber \(X\in\F{\cconf{C}}\).
The strict algebraic automorphisms of \(\cconf{C}\) form a group, which we denote by~\(\salgaut{\cconf{C}}\).
Now, we get the following reformulation of separability:
\begin{lemma}\label{Lemma: reduction to strict algebraic automorphisms}
A star-free and \(2\)-induced \(k\)-ary coherent configuration \(\cconf{C}\)
with \(5\)-bounded fibers
is separable if and only if every strict algebraic automorphism is induced by
a combinatorial automorphism.
\begin{proof}
The forward implication is immediate.
For the converse implication,
let \(f\colon\cconf{C}\to\mathcal{D}\) be an algebraic isomorphism.
By Lemma~\ref{Lemma: Every algebraic isomorphism is induced by combinatorial isomorphism on every fiber}, we find a combinatorial isomorphism
\(\phi\colon V(\cconf{C})\to V(\mathcal{D})\)
which satisfies \(\phi(X)=f(X)\) for every fiber \(X\in\F{\cconf{C}}\).
Because \(\cconf{C}\) is \(2\)-induced, this map is also a combinatorial isomorphism
between~\(\cconf{C}\) and~\(\mathcal{D}\) and thus induces an algebraic isomorphism
\(g_\phi\colon\cconf{C}\to\mathcal{D}\) such that \(f(X)=g_\phi(X)\) for every
fiber \(X\in\cconf{C}\).
Now, \(g_\phi^{-1}\circ f\) is a strict algebraic automorphism of~\(\cconf{C}\)
and is thus induced by some combinatorial automorphism~\(\psi\).
But then \(f\) is induced by \(\phi\circ\psi\), which proves separability of~\(\cconf{C}\).
\end{proof}
\end{lemma}
\subsection{Strict algebraic automorphisms}\label{Section: strict algebraic automorphisms}
In this section, we give a polynomial time algorithm to decide
whether every strict algebraic automorphism of a
\(k\)-ary coherent configuration \(\cconf{C}\) is induced by a combinatorial automorphism.

We will heavily make use of the fact that the graph isomorphism problem
is polynomial-time solvable for graphs with bounded color classes:
\begin{lemma}[{\cite{GIBoundedCCSLasVegas,GIBoundedCCS}}]\label{Lemma: GI fpt in color multiplicity}
Isomorphism of \(k\)-ary relational structures of order \(n\) and \(c\)-bounded color classes can be decided
in time \(O_{k,c}(n^{O(k)})\). Moreover, a generating set of \(\Aut(\relstruc{A})\)
for of a \(k\)-ary relational structure \(\relstruc{A}\) of order \(n\) and \(c\)-bounded color classes can also be computed
in the same time bounds.
\end{lemma}

\noindent By encoding the algebraic structure of a \(k\)-ary coherent configuration into a graph \(G_{\cconf{C}}\) such that
the strict algebraic automorphisms of \(\cconf{C}\) become automorphisms of \(G_{\cconf{C}}\), we obtain the following: 
\begin{lemma}\label{Lemma: computing generating set of A(C)}
There is an algorithm running in time \(O_{k,c}(n^{O(k)})\) that,
given a \(k\)-ary coherent configuration \(\cconf{C}\) of order \(n\) and \(c\)-bounded fibers,
computes a generating set of \(\salgaut{\cconf{C}}\).
\begin{proof}
We construct a graph \(G_{\cconf{C}}\) with bounded color classes and order \(O(n^k)\)
whose automorphism group is isomorphic to \(\salgaut{\cconf{C}}\).
Lemma~\ref{Lemma: GI fpt in color multiplicity} then yields the claim.

For this, we call a \(k\)-tuple \(\vec{R}=(R_1,\dots,R_k)\) of basis relations of \(\cconf{C}\) \emph{compatible} if for some basis relation \(R\in\cconf{C}\) we have
\(\intnum{R}{R_1,\dots,R_k}>0\). Let \(\cconf{C}^{(k)}\) be the collection of all such tuples.
Because for every compatible tuple \(\vec{R}\) of basis relations there must exist some
\(\vec{x}\in V(\cconf{C})^k\) and \(x'\in V(\cconf{C})\) such that
\(\vec{x}\frac{x'}{i}\in R_i\) for all \(i\in[k]\), there are at most \(n^{k+1}\) compatible tuples,
and these can be found in time \(n^{O(k)}\) by enumerating all such pairs \((\vec{x},x')\)
noting which compatible tuple of basis relations they correspond to.

Now, we define the graph \(G_{\cconf{C}}\)
on the vertex set \(V(G_{\cconf{C}})\coloneqq\cconf{C}\dotcup\cconf{C}^{(k)}\)
by putting in the following labeled edges:
\begin{itemize}
\item connect \(R\) and \(\vec{R}\) with an edge labeled \(\intnum{R}{\vec{R}}\),
\item connect \(R\) and \(\vec{R}\) with an edge labeled \(\{i\in[k]\colon R_i=R\}\),
\end{itemize}
Finally, we define a vertex coloring on \(G_{\cconf{C}}\):
\begin{itemize}
\item color each basis relation \(R\) using the tuple \((R_{\{i\}})_{i\in[k]}\) of fibers of its components,
\item color each compatible tuple of basis relations \(\vec{R}\) using the tuple of colors we assigned to its components. 
\end{itemize}
This graph can clearly be constructed in time \(n^{O(k)}\) and has order \(n^{O(k)}\).

Because the fibers of \(\cconf{C}\) have order at most \(c\), there can be at most \(c^k\) basis relations
sharing the same color. Thus, the number of compatible tuples of basis relations sharing a color is bounded
by \((c^k)^k=c^{k^2}\in O_{k,c}(1)\). 
Moreover, its automorphism group is as required: each automorphism is determined by how it permutes
the basis relations, the edge-colors ensure that this automorphism is an algebraic automorphism,
and the vertex-colors ensure strictness.
\end{proof}
\end{lemma}

As the collection of those strict algebraic automorphisms which are induced by combinatorial ones
forms a subgroup of \(\salgaut{\cconf{C}}\), it now suffices to check whether
each of the polynomially many elements of the generating set of \(\salgaut{\cconf{C}}\)
is induced by a combinatorial automorphism. This can also be checked in polynomial time:
\begin{lemma}\label{Lemma: Deciding whether algebraic iso is combinatorial in polynomial time}
There is an algorithm running in time \(O_{k,c}(n^{O(k)})\) that,
given a \(k\)-ary coherent configuration \(\cconf{C}\) with \(c\)-bounded fibers,
and a strict algebraic automorphism \(f\in\salgaut{\cconf{C}}\),
outputs whether \(f\) is induced by a combinatorial automorphism.
\begin{proof}
Let \(\relstruc{C}\) be an arbitrary colored variant of \(\cconf{C}\) and \(\relstruc{C}^f\)
another colored variant such that~\(f\) becomes a color-preserving map
between the \(k\)-ary color classes of these two structures.
Combinatorial automorphisms inducing \(f\) now naturally correspond to isomorphisms between~\(\relstruc{C}\)
and~\(\relstruc{C}^f\), meaning that \(f\) is induced by a combinatorial automorphism if and only if
\(\relstruc{C}\simeqq\relstruc{C}^f\). As \(\cconf{C}\) has bounded fibers,
so does \(\relstruc{C}\).
Thus, we can decide the latter in the required time by Lemma~\ref{Lemma: GI fpt in color multiplicity}.
\end{proof}
\end{lemma}

\begin{corollary}\label{Corollary: quasiseparability in polynomial time}
There is an algorithm running in time \(O_{k,c}(n^{O(k)})\) that,
given a \(k\)-ary coherent configuration \(\cconf{C}\) with \(c\)-bounded fibers,
either
\begin{enumerate}
\item outputs a \(\phi\in\salgaut{\cconf{C}}\) which not induced by a combinatorial automorphism, or
\item correctly returns that no such automorphisms exist.
\end{enumerate}
\end{corollary}

\noindent This allows us to finally prove our first main theorem:
\WLkIDccsFive*
\begin{proof}
	In a first step, we run \(\WL[k]\) on \(G\) to get the \(2\)-induced configuration \(\cconf{C}\coloneqq\CC{k}(G)\).
	By Lemma~\ref{Lemma: identification and separability}, it remains to decide whether \(\cconf{C}\) is separable.
	Now, we eliminate disjoint unions of stars using Lemma~\ref{Lemma: Elimination of disjoint unions of stars I},
	and Lemma~\ref{Lemma: Elimination of disjoint unions of stars II}, while maintaining \(2\)-inducedness of \(\cconf{C}\).
	By Lemma~\ref{Lemma: reduction to strict algebraic automorphisms}, it remains to decide
	whether every strict algebraic automorphism is induced by a combinatorial one.
	This can be achieved using Corollary~\ref{Corollary: quasiseparability in polynomial time}.
	
	If this is the case, the input structure is identified by \WL[k].
	Otherwise, we obtain a strict algebraic automorphism \(f\) which is not induced by
	a combinatorial automorphism. By adding back all interspaces containing a disjoint union of stars,
	we can extend \(f\) to an algebraic isomorphism \(\widehat{f}\colon\CC{k}(G)\to\cconf{D}\)
	which is not induced by a combinatorial isomorphism.
	But then, we can obtain a witnessing graph \(H\) from \(G\) by replacing its edge set by its \(\widehat{f}|_2\)-image
	and similarly translating vertex- and edge-colors along~\(\widehat{f}\).
\end{proof}

\section{Identification for structures with bounded abelian color classes}
The approach we used in Section~\ref{Section: deciding identification for graphs of color multiplicity 5}
to decide the \WL[k]-identification problem for graphs with \(5\)-bounded color classes
does not easily generalize to graphs with larger color classes or to relational structures of higher arity.
In particular, Lemma~\ref{Lemma: reduction to strict algebraic automorphisms} was crucial in the reduction
of \WL[k]-identification to a statement on strict algebraic automorphisms which could be handled
using group-theoretic techniques.
The proof of the lemma was based on an explicit case distinction on the possible
isomorphism types of interspaces, and fails for graphs with larger color classes.
In this section, we show that Lemma~\ref{Lemma: reduction to strict algebraic automorphisms}
remains true in the special case of relational structures with bounded abelian color classes, i.e.,
structures for which the automorphism group of the structure induced on each color class is abelian.
Such structures were already considered in the context of descriptive complexity theory~\cite{abelian_color_classes}
and include both CFI-graphs \cite{CFI} and multipedes \cite{multipedesI, multipedesII} on ordered base graphs.

\subsection{Coherent configurations with abelian fibers}
To start, we translate the concept of abelian color classes to the corresponding concept of abelian fibers
for \(k\)-ary coherent configurations.

A combinatorial automorphism $\phi$ of a \(k\)-ary coherent configuration \(\cconf{D}\)
is \emph{color-preserving} if $\phi$ fixes every basis relation of $\cconf{D}$.
This is equivalent to $\phi$ being an automorphism of every colored variant of $\cconf{D}$
or to the algebraic automorphism induced by $\phi$ being the identity
(recall that combinatorial automorphisms are not required to fix every basis relation, but only the partition of \(V(\cconf{D})^k\) into basis relations).
We say that a coherent configuration~\(\cconf{C}\) has \emph{abelian fibers} if, for each fiber \(X\in\F{\cconf{C}}\),
the group of color-preserving combinatorial automorphisms of \(\cconf{C}[X]\) is abelian.
\begin{lemma}\label{lem:abelian_colors:abelian_colors_implies_abelian_fibers}
Let \(\relstruc{A}\) be a relational structure of arity at most \(k\).
If \(\relstruc{A}\) has abelian color classes, then \(\CC{k}(\relstruc{A})\) has abelian fibers.
\end{lemma}

Recall that an algebraic automorphism of a \(k\)-ary coherent configuration \(\cconf{C}\)
is strict if it fixes every basis relation inside a fiber.
Similarly, we call a combinatorial automorphism \(\phi\) \emph{strict} if it fixes every basis relation inside a fiber,
i.e., if it fixes every fiber and its restriction to every fiber is color-preserving.
This is equivalent to \(\phi\) inducing a strict algebraic automorphism.

Towards understanding the structure of abelian fibers, we start with one simple group-theoretic observation.
\begin{lemma}\label{lem:abelian_colors:abelian+transitive->regular}
Let \(\Gamma\subseteq\Sym(\Omega)\) be an abelian group acting transitively on a set
\(\Omega\). Then the group action is regular.
\begin{proof}
Because \(\Gamma\) acts transitively, the stabilizer subgroup of all elements in \(\Omega\) are pairwise conjugated.
But because \(\Gamma\) is abelian, this implies that they are equal. Thus, every permutation
that stabilizes some element already stabilizes all, which is only true for the identity.
\end{proof}
\end{lemma}

In order to apply this observation to abelian fibers, we need one more definition.
For a fiber \(X\in\F{\cconf{C}}\), a binary basis relation \(S\in\cconf{C}|_2[X]\) is called \emph{thin}
if every vertex in \(X\) is incident to exactly one ingoing and exactly one outgoing \(S\)-edge, that is, if \(S\) is either a matching
or a union of directed cycles.
The fiber \(X\) is called \emph{thin} if all basis relations \(R\in\cconf{C}|_2[X]\) are thin and if this is true for all fibers
of \(\cconf{C}\), we say that \(\cconf{C}\) has \emph{thin fibers}.

\begin{corollary}\label{lem:abelian_colors:small_abelian_fibers_are_thin}
Let \(\cconf{C}\) be a \(k\)-ary coherent configuration.
Then every abelian fiber of order at most \(k\) is thin.
\begin{proof}
Let \(X\in\F{\cconf{C}}\) be an abelian fiber of order at most \(k\).
Then \(\cconf{C}|_2[X]\) is the partition of \(X^2\) into orbits
under the natural action of the group of color-preserving automorphisms.

Pick some \(x\in X\) and assume that some binary basis relation \(S\in\cconf{C}|_2[X]\)
contains two pairs \(xy\) and \(xy'\) for \(y,y'\in X\).
But this implies that there is a color-preserving automorphism \(\phi\) of \(\cconf{C}[X]\) that maps
\(xy\) to \(xy'\). But as the group of color-preserving automorphism of \(\cconf{C}[X]\) is abelian and acts transitvely,
\(\phi(x)=x\) implies \(y'=\phi(y)=y\). Thus, the basis relation \(S\) is thin.
\end{proof}
\end{corollary}

Finally, we need one well-known lemma on the structure of thin fibers,
which essentially states that thin fibers correspond to Cayley graphs of
their automorphism groups.
\begin{lemma}[{\cite[Section 2.1.4]{CC}}]\label{lem:abelian_colors:structure_of_thin_fibers}
Let \(\cconf{C}\) be a \(2\)-ary coherent configuration on a single thin fiber.
Then the basis relations of \(\cconf{C}\) are precisely those of the form
\(S_\phi\coloneqq\{x\phi(x)\colon x\in V(\cconf{C})\}\) for color-preserving combinatorial automorphisms
\(\phi\) of \(\cconf{C}\).
\end{lemma}

\subsection{Separability of configurations with bounded thin fibers}
Let \(\cconf{C}\) be a \(k\)-ary coherent configuration.
Recall that in Section~\ref{sec:stars}, we defined for every functional basis relation,
that is, a basis relation with out-degree at most \(1\) at every vertex of \(\cconf{C}\),
a map
\begin{align*}
\nu_S\colon V(\cconf{C}) &\to V(\cconf{C}),\\
v&\mapsto
	\begin{cases}
		w &\text{if } vw\notin S,\\
		v &\text{if no such } w \text{ exists.}
	\end{cases}
\end{align*}
and, for  \(I\subseteq [k]\), further the maps \(\nu_S^I\colon V(\cconf{C})^k\to V(\cconf{C})^k\)
which act as \(\nu_S\) on all components in \(I\) and as the identity on all components not in \(I\).

Because every thin basis relation \(S\) lies within a single fiber, every vertex of that fiber
also has in-degree exactly \(1\) with respect to \(S\).
This means that not only is every thin basis relation functional, but the maps \(\nu_S\)
are bijective in this case.
Thus, the following lemma is an immediate consequence of Lemma~\ref{lem:stars:definable_maps_are_algebraic}.
\begin{lemma}\label{lem:abelian_colors:thin_relations_define_algebraic_map}
For every \(k\)-ary coherent configuration \(\cconf{C}\), every thin basis relation \(S\in\cconf{C}|_2[X]\),
and every \(I\subseteq [k]\), the map \(\nu_S^{I}\) induces an algebraic automorphism of \(\cconf{C}\).
Furthermore, for every algebraic isomorphism
\(f\colon\cconf{C}\to\cconf{D}\), we get \(f\circ\nu_S^{I}=\nu_{f(S)}^{I}\circ f\).
\end{lemma}

Next, we show that \(k\)-ary coherent configurations with few, thin fibers are separable:
\begin{lemma}\label{lem:abelian_colors:one_pebble_suffices}
Let \(\cconf{C}\) be a \(k\)-ary coherent configuration with at most \(k\) fibers.
If \(\cconf{C}\) has thin fibers, then \(\cconf{C}\) is separable.
\begin{proof}
Let \(\relstruc{C}\) be a colored variant of \(\cconf{C}\),
where we additionally also add all thin basis relations within each fiber
as binary basis relations inside the color classes.
Then \(\cconf{C}=\CC{k}(\relstruc{C})\), which, using Lemma~\ref{Lemma: identification and separability},
implies that \(\cconf{C}\) is separable if and only if \(\relstruc{C}\) is identified by \WL[k].

Thus, assume \(\relstruc{C}\equiv\relstruc{D}\). We argue using the bijective \((k+1)\)-pebble game
that \(\relstruc{C}\simeqq\relstruc{D}\).
We start by considering a winning position \(x\mapsto y\) for Duplicator
in the bijective \((k+1)\)-pebble game between \(\relstruc{C}\) and \(\relstruc{D}\)
with only a single placed pebble pair. Assume Spoiler picks up a second pebble pair
and let \(f\colon V(\relstruc{C})\to V(\relstruc{D})\) be the bijection
that Duplicator provides according to their winning strategy.
Note that for every thin basis relation \(S\) in the color class of \(x\),
the map \(f\) must map the unique \(S\)-neighbor of \(x\) to the unique \(S\)-neighbor of~\(y\),
which completely determines the map~\(f\) on the color class of~\(x\).
If this is not the case, Spoiler can place pebbles on this vertex pair and wins in the next round.
Further note that for every different vertex \(x'\) in the color class of \(x'\),
this same bijection on the color class must also be winning in position
\(xx'\mapsto yf(x')\) and thus also in position \(x'\mapsto f(x')\).
Thus, as long as Spoiler never picks up all pebbles from this color class,
the bijections provided by Duplicator stay fixed on this color class.

Now, assume Spoiler places a pebble in every color class of \(\relstruc{C}\)
and thus reaches a position \(\vec{x}\mapsto \vec{y}\), which is still
winning for Duplicator. Again, let \(f\) be the bijection
that Duplicator provides when Spoiler picks the last remaining pebble pair.

We argue that \(f\) is an isomorphism, by showing that
for every \(k\)-tuple of vertices \(\vec{z}\in V(\relstruc{C}^k)\),
Spoiler can force the game to reach the position \(\vec{z}\mapsto f(\vec{z})\).
Because Duplicator has a winning strategy, these positions
must be partial isomorphisms, which then implies that \(f\) is indeed an isomorphism.

By our previous remarks, it suffices to observe that Spoiler can clearly pebble
all vertices in \(\vec{z}\) without ever removing all pebbles
from some color class that contains some vertex from~\(\vec{z}\).
This way, the bijections provided by Duplicator must always agree with \(f\)
on all color classes containing a vertex from \(\vec{z}\), which proves the claim.
\end{proof}
\end{lemma}

Finally, we are ready to once again reduce the question of separability to only strict algebraic
automorphisms, which we can again deal with using Corollary~\ref{Corollary: quasiseparability in polynomial time}.
\begin{lemma}\label{lem:abelian_colors:autoseparable}
Let \(\cconf{C}\) be a \(k\)-ary coherent configurations with thin fibers.
Then \(\cconf{C}\) is separable if and only if every strict algebraic automorphism of \(\cconf{C}\)
is induced by a combinatorial automorphism.
\begin{proof}
The forward implication is immediate, so it suffices to show the backward implication.
So assume that every strict algebraic automorphism of \(\cconf{C}\) is induced by a combinatorial automorphism
and let \(f\colon\cconf{C}\to\cconf{D}\) be an algebraic isomorphism. We need to show that \(f\) is induced by a combinatorial isomorphism.

By Lemma~\ref{lem:abelian_colors:one_pebble_suffices},
\(f\) is induced by a combinatorial isomorphism \(\psi_{\vec{X}}\colon\cconf{C}[\vec{X}]\to\cconf{D}[f(\vec{X})]\) on every union \(\vec{X}=X_1\cup\dots\cup X_k\) of \(k\) fibers
and thus in particular by a combinatorial isomorphism \(\phi_X\colon\cconf{C}[X]\to\cconf{D}[f(X)]\)
on every single fiber.
We define a bijection \(\phi\colon V(\cconf{C})\to V(\cconf{D})\) by setting \(\phi|_X\coloneqq\phi_X\)
for every fiber \(X\in\F{\cconf{C}}\) and claim that \(\phi\) is a combinatorial isomorphism that
induces \(f\) on every fiber.

Clearly, \(\phi|_X\) is a combinatorial isomorphism inducing \(f\) for every fiber \(X\in\F{\cconf{C}}\),
hence it remains to show that \(\phi\) is also a combinatorial isomorphism on the whole configuration~\(\cconf{C}\).
Because every basis relation \(R\in\cconf{C}\) is contained
in a subconfiguration \(\cconf{C}[\vec{X}]\) induced on the union of \(k\) fibers
\(\vec{X}=X_1\cup\dots \cup X_k\), it suffices to show that \(\phi|_\vec{X}\colon\vec{X}\to f(\vec{X})\)
is a combinatorial isomorphism from \(\cconf{C}[\vec{X}]\) to \(\cconf{D}[f(\vec{X})]\) for every
such \(\vec{X}\).

For this, we first note that the map \(\psi_\vec{X}\) is such a combinatorial isomorphism,
which implies that for every combinatorial automorphism \(\theta\) of \(\cconf{C}[\vec{X}]\),
the map \(\psi_\vec{X}\circ\theta\) is such a combinatorial isomorphism as well.
Thus, it would suffice to prove that the permutation \(\psi_\vec{X}^{-1}\circ\phi|_\vec{X}\)
is a combinatorial automorphism of \(\cconf{C}[X]\).
For this, note since both \(\psi_\vec{X}\) and \(\phi|_\vec{X}\) induce \(f\) on every fiber,
the composition \(\psi_\vec{X}^{-1}\circ\phi|_\vec{X}\) is a color-preserving automorphism
restricted to every fiber \(X_i\subseteq \vec{X}\).

By Lemma~\ref{lem:abelian_colors:structure_of_thin_fibers}, this implies
that for every fiber \(X_i\subseteq\vec{X}\), there is a thin basis relation
\(S_i\in\cconf{C}|_2[X_i]\) such that \(\psi_\vec{X}^{-1}\circ\phi|_{X_i}=\nu_{S_i}|_{X_i}\).
Thus, we can write
\[\psi_\vec{X}^{-1}\circ\phi|_{\vec{X}}=\prod_{i=1}^k \nu_{S_i},\]
which is a composition of combinatorial automorphisms and thus also a combinatorial automorphism.
Thus, we get that
\(\phi|_\vec{X}=\psi_\vec{X}\circ\left(\psi_\vec{X}^{-1}\circ\phi|_\vec{X}\right)\)
is indeed a combinatorial isomorphism for every union of \(k\) fibers \(\vec{X}\), which implies
that \(\phi\) is a combinatorial isomorphism which induces \(f\) on every fiber.

Finally, it follows that \(\phi^{-1}\circ f\) is a strict algebraic automorphism
which, by assumption, is induced by a combinatorial automorphism \(\theta\).
But then, \(\phi\circ\theta\) induces \(f\).
\end{proof}
\end{lemma}

\IdentificationAbelianColors*
\begin{proof}
Let \(\relstruc{A}\) be a relational structure of arity \(r\).
Then \(\relstruc{A}\) is identified by \(\WL[k]\) if and only if
\(\CC{k}(\relstruc{A})\) is separable.
Because \(\CC{k}(\relstruc{A})\) has \(c\)-bounded thin fibers by Lemma~\ref{lem:abelian_colors:abelian_colors_implies_abelian_fibers} and Lemma~\ref{lem:abelian_colors:small_abelian_fibers_are_thin},
Lemma~\ref{lem:abelian_colors:autoseparable} implies
that separability of \(\CC{k}(\relstruc{A})\) is equivalent to every strict
algebraic automorphism of \(\CC{k}(\relstruc{A})\) being induced by a combinatorial automorphism.
This can be checked in the given time using Corollary~\ref{Corollary: quasiseparability in polynomial time},
and in case of a negative answer, we can construct a non-isomorphic but non-distinguished
structure from the strict algebraic automorphism not induced by a combinatorial one
as in Theorem~\ref{Theorem: WL-identification on ccs 5}.
\end{proof}

Note that the restriction to relational structures of arity at most \(k\) is insubstantial,
because the standard variant of the Weisfeiler-Leman algorithm given in Section~\ref{sec:preliminaries}
does not identify any relational structure of arity larger than \(k\),
simply because it does not consider tuples of length larger than \(k\) and thus
cannot even detect whether a relation of arity larger than \(k\) is empty.
While there are variants of \WL[k] which identify some \((k+1)\)-ary relational structures,
these variants can be treated similarly to decide identification by those algorithms.

\section{Hardness}
In this section, we prove hardness results that complement the positive results in the previous two sections.
We start by showing that, when the dimension \(k\) is considered part of the input,
the \WL[k]-equivalence problem and the \WL[k]-identification problem are \coNP-hard
and \NP-hard, respectively.
This is achieved via reductions from \textsc{Tree-width}, which is \NP-hard even
over cubic graphs~\cite{twNPhard}. The reduction is based on the CFI-construction, see Section~\ref{Section: CFI}.

\WLdimNPhard*
\begin{proof}
\textsc{Tree-width} is the problem to decide
whether the tree-width of a given graph \(G\) is at most a given number \(k\).
This problem is \NP-hard even over cubic graphs~\cite{twNPhard}.
By Lemma~\ref{Lemma: CFI vs. tw}, we have \(\tw(G)=\WLdim(\CFI(G,0))\). Thus, computing the tree-width of a cubic graph~\(G\)
reduces to computing the Weisfeiler-Leman dimension of its CFI-graph. Because the CFI-graphs of cubic graphs
have \(4\)-bounded color classes, and CFI-graphs of cubic graphs can be efficiently computed, the hardness
result for graphs with \(4\)-bounded color classes follows.

The claim for the class of simple graphs follows from the observation that we can encode colors into gadgets
we attach to every vertex, and that these gadgets can be chosen such that they do not affect the Weisfeiler-Leman dimension.
\end{proof}
Because in the above proof, we could explicitly construct a non-isomorphic but equivalent graph,
this also yields \coNP-hardness of the \WL[k]-equivalence problem,
which was also independently observed in \cite{seppelt_WLEquivcoNPhard}.
\begin{theorem}
The problem of deciding, for a given pair of graph \(G\) and \(H\) and a natural number \(k\geq 1\),
whether \(G\equiv_{\WL[k]} H\), is \coNP-hard.
\end{theorem}

Now, we once again turn to the \WL[k]-identification problem for a fixed dimension \(k\geq 2\),
and show that both over uncolored simple graphs, and over
simple graphs with \(4\)-bounded color classes, the problem is \Ptime-hard under logspace-uniform \ACZ-reductions.

We reduce from the \Ptime-hard monotone circuit value problem MCVP~\cite{MCVP}.
Our construction of a graph from a monotone circuit closely resembles
the reductions of Grohe~\cite{WLEquivalencePHard} to show \Ptime-hardness of the \WL[k]-equivalence problem.
A similar reduction was also used to prove \Ptime-hardness of the identification problem
for the color refinement algorithm (\WL[1])~\cite{CRIdentificationPHard}.

The reduction based one so-called \emph{one-way switches},
which were introduced by Grohe~\cite{WLEquivalencePHard}.
These graph gadgets allow color information
computed by the Weisfeiler-Leman algorithm to pass in one direction,
but block it from passing in the other.
And while Grohe provides one-way switches for every dimension of the Weisfeiler-Leman algorithm,
his gadgets have large color classes and are difficult to analyze.
Instead, we give a new construction of such gadgets with \(4\)-bounded color classes
based on the CFI-construction. We then use these one-way switches
to construct a graph from an instance of the monotone circuit value problem from
the identification of which we can read off the answer to the initial MCVP-query.

\subsection{One-way switches}
\label{Section: One-way switches}
In the following sections, we fix a dimension \(k\geq 2\) of the Weisfeiler-Leman algorithm.
A \(k\)-one-way switch is a graph gadget with a pair of \emph{input vertices} \(\{y_1,y_2\}\),
and a pair of output vertices \(\{x_1,x_2\}\), which each form a color class of size \(2\).
We say that \emph{a pair of vertices is split} if the two vertices are colored differently.
We say that \WL[k] splits a pair if the coloring computed by \WL[k] splits the pair.

The fundamental property of the \(k\)-one-way switch is the following:
whenever the input pair \(\{y_1,y_2\}\) of the one-way switch is split, \WL[k] also splits the output pair,
but not the other way around. One-way switches thus only allow one-way flow of \WL[k]-color information.
In contrast to Grohe's gadgets, our one-way switches are based on the CFI-construction, see Section~\ref{Section: CFI}.

We start by defining our base graph. Consider a wall graph consisting of \(k-1\) rows of \(k\) bricks each.
Then, we attach a new vertex \(v\) to the two upper corner vertices of the first row. The resulting graph
\(B_k\) is depicted in Figure~\ref{Figure: wall graph with input vertex}.
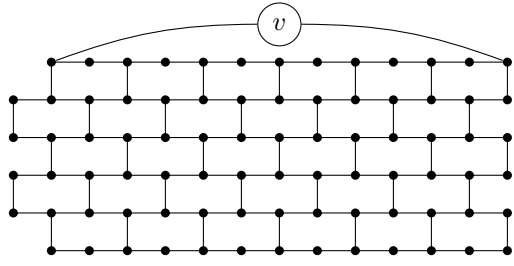
\begin{figure}
\centering
\def\height{6}
\def\width{14}
\begin{tikzpicture}[scale=0.5]
\foreach \y [evaluate=\y as \ypp using \y+1] in {1,...,\height}{
	\foreach \x [evaluate=\x as \s using \x+\y] in {1,...,\width}{
		\ifthenelse{\y>1 \AND \y<\height \OR \x>1}{
			\node[shape=circle,fill=black,inner sep=1.2pt] (\x;\y) at (\x,\y) {};
		}{}
	}
	\ifthenelse{\y=1 \OR \y=\height}{
			\draw (2,\y) -- (\width,\y);
		}{
			\draw (1,\y) -- (\width,\y);		
		}
	
}
\node (1;1) at (-1,-1) {};
\foreach \y [evaluate=\y as \ymm using \y-1] in {2,...,\height}{
	\foreach \x [evaluate=\x as \s using \x+\y] in {1,...,\width}{
		\ifodd\s
		\else
			\draw (\x,\ymm) -- (\x,\y);
		\fi
	}
}

\node[shape=circle, draw=black] (input) at (\width/2+1,\height+1) {\(v\)};
\draw (2;\height) to[bend left=10] (input);
\draw (input) to[bend left=10] (\width;\height);

\end{tikzpicture}
\caption{The base graph \(B_6\) of the CFI-graphs underlying our one-way switches.}
\label{Figure: wall graph with input vertex}
\end{figure}
\begin{lemma}\label{Lemma: Base graph of ows has tree-width k+1}
The graph \(B_k\) has tree-width \(k+1\), while \(B_k-v\) has tree-width \(k\).
\begin{proof}
Because the tree-width of a graph is invariant under subdividing edges,
we can delete~\(v\) and directly connect the two upper corner vertices of \(B_k\).
Let~\(B_k'\) be the resulting graph.

For every positive integer $\ell$, let \(G_{k,\ell}'\) be the graph obtained from a \(k\times\ell\)-grid graph by directly connecting two corners
that lie on a common side of length \(\ell\).
Then \(G_{k,k+1}'\) is a minor of \(B_k'\) and thus \(\tw(B_k')\geq \tw(G_{k,k+1}')\).
Similarly, up to a subdivision of some edges, which does not affect the tree-width,
\(B_k'\) is a subgraph of \(G_{k,2k+1}'\). Thus,
\[\tw(G_{k,k+1}')\leq\tw(B_k')\leq\tw(G_{k,2k+1}').\]
It thus suffices to argue that \(\tw(G_{k,\ell}')=k+1\) for all \(\ell>k\).
Because adding a single vertex or edge to a graph increases the tree-width by at most \(1\),
and the \(k\times\ell\)-grid graph has tree-width \(\min(k,\ell)\), we have
\(\tw(B_k-v)=k\) and \(\tw(G_{k,\ell}')\leq k+1\).

For the matching lower bound, we use the theory of brambles \cite{brambles}.
Recall that a bramble~\(\mathcal{B}\) of a finite graph~\(H\) is a collection of connected subgraphs
of~\(H\) such that every two subgraphs in~\(\mathcal{B}\) either overlap, or are connected by an edge.
The order of a bramble is the size of a minimal hitting set, i.e., the smallest \(k\in\N\)
such that there exist \(k\) vertices \(v_1,\dots,v_k\) such that every
set \(B\in\mathcal{B}\) contains one of the vertices \(v_i\).
It turns out that the maximal order of a bramble in \(H\) is precisely one more
than the tree-width of \(H\)~\cite{brambles}.

For the matching lower bound, it thus suffices to construct a bramble in \(G_{k,\ell}'\) of order at least \(k+2\).
Let the vertex set of \(G_{k,\ell}'\) be \([k]\times[\ell]\),
with the extra edge connecting \((1,1)\) to \((1,\ell)\).
For each \(i\in[k]\setminus\{k\}\) and \(j\in[\ell]\setminus\{\ell\}\),
we define a set
\[U_{i,j}\coloneqq\{(a,b)\in([k]\setminus\{k\})\times([\ell]\setminus\{\ell\})\colon a=i \text{ or } b=j\}\]
to be the union of the \(i\)-th row and \(j\)-th column, while excluding
the vertices from the last row or column.
Further, we take the two sets
\[R_k\coloneqq\{k\}\times[\ell]\quad\text{and}\quad C_\ell\coloneqq ([k]\setminus\{1,k\})\times\{\ell\}.\]
Finally, we define our bramble to be
\begin{align*}
&\bigl\{R_k,C_\ell\bigr\}\\
&{}\cup\bigl\{U_{i,j}\colon i\in[k]\setminus\{1,k\}, j\in[\ell]\setminus\{\ell\}\bigr\}\\
&{}\cup\Bigl\{\bigl(U_{1,j}\cup\{(1,\ell)\}\bigr)\setminus\{(1,i)\}\colon j\in[\ell]\setminus\{\ell\},i\in[\ell]\setminus\{j\}\bigr\}.
\end{align*}

To argue that this bramble has order at least \(k+2\), we need to show that it does not admit a hitting set of order at most \(k+1\),
that is, we must argue that for every choice of \(k+1\) vertices there exists some set in the bramble which contains none of the vertices.
Assume for contradiction that \(H\subseteq V(G_{k,\ell}')\) was such a hitting set.
Because the sets \(R_k\) and \(C_\ell\) are disjoint from all other sets, \(H\) must contain at least one vertex in each of these sets.
Hence, only \(k-1\) vertices remain to hit all other sets in the bramble. These \(k-1\) vertices must miss one of the first \(\ell-1\) columns,
say column \(j\). If they also miss one of the first \(k-1\) rows, say row \(i\), then they miss the set \(U_{i,j}\).
Otherwise, there is precisely one vertex per row. Let \((1,i)\) be the vertex in the first row.
But then the set misses \(\bigl(U_{1,j}\cup\{(1,\ell)\}\bigr)\setminus\{(1,i)\}\).
\end{proof}
\end{lemma}

\noindent Now, we are ready to construct our one-way switches.
\newcounter{enumcounter}
\begin{lemma}[{compare \cite[Lemma 14]{WLEquivalencePHard}}]\label{Lemma: Existence of one-way switches}
For every \(k\geq 2\), there is a colored graph \(O^k\) with \(4\)-bounded color classes,
called \emph{\(k\)-one-way switch}, with an input pair \(\{y_1,y_2\}\) and an output pair \(\{x_1,x_2\}\)
satisfying the following properties:
\begin{enumerate}
\item \label{Property: split ows identified}
	The graph \(O^k_{\spl}\) obtained by splitting the input pair \(\{y_1,y_2\}\) is identified by \WL[k].
\item \label{Property: split ows splits output pair}
	\WL[k] splits the output pair \(\{x_1,x_2\}\) of \(O^k_{\spl}\).
\item \label{Property: output vertices in different orbits}
	There is no automorphism of \(O^k\) exchanging the output vertices \(x_1\) and \(x_2\).
	\setcounter{enumcounter}{\value{enumi}}
\end{enumerate}
Furthermore, there are sets of positions in the bijective \((k+1)\)-pebble game between \(O^k\)
and itself, called \emph{trapped} and \emph{twisted} such that
\begin{enumerate}
	\setcounter{enumi}{\value{enumcounter}}
\item \label{Property: Trapped and twisted partial isos}
	every trapped or twisted position is a partial isomorphism,
\item \label{Property: Trapped and twisted winning}
	Duplicator can avoid non-trapped positions from trapped ones and non-twisted positions from twisted ones,
\item \label{Property: Trapped compatible with output pair}
	for every trapped position \(\vec{a}\mapsto\vec{b}\), the position \(\vec{a}x_1\mapsto\vec{b}x_1\)
	is also trapped,\footnote{ If the position \(\vec{a}x_1\mapsto\vec{b}x_1\) contains more than \(k+1\) pebbles, this means that every subposition on at most \(k+1\) pebbles is trapped}
\item \label{Property: Twisted compatible with output pair}
	for every twisted position \(\vec{a}\mapsto\vec{b}\), the position \(\vec{a}x_1\mapsto\vec{b}x_2\) is also twisted,
\item \label{Property: Switching input is trapped and twisted}
	the positions \(y_1y_2\mapsto y_1y_2\) and \(y_1y_2\mapsto y_2y_1\) are both trapped and twisted,
\item \label{Property: Subpositions of trapped and twisted}
	every subposition of a trapped position is trapped, and every subposition of a twisted position is twisted
\end{enumerate}
\begin{proof}
Let \(O^k\) be the (untwisted) CFI-graph of \(B_k\), but with a CFI-gadget of degree \(3\) added for the vertex \(v\)
instead of a gadget of degree \(2\).
This leaves one outer pair of this gadget free
which we use as our output pair \(\{x_1,x_2\}\).
Furthermore, we use one of the other two outer pairs of this same CFI-gadget
as the input pair \(\{y_1,y_2\}\).

Now, if we fix the output pair \(\{x_1,x_2\}\) by individualizing one of the two vertices,
the resulting graph corresponds to the usual CFI-graph of \(H\),
while switching the pair \(\{x_1,x_2\}\) corresponds to the twisted CFI-graph of \(H\). In particular,
as these graphs are not isomorphic, there is no automorphism of \(O^k\) switching the pair \(\{x_1,x_2\}\),
which proves Property~\ref{Property: output vertices in different orbits}.

Moreover, splitting the input pair \(\{y_1,y_2\}\) has the same effect to the power of \WL[k]
as removing one of the two edges incident to \(v\) in the base graph \(B_k\) has.
When removing this edge in the base graph, the resulting graph is essentially equivalent
to the CFI-graph of the \(k\times(k+1)\)-wall graph with one corner vertex replaced by a CFI-gadget
of degree~\(3\) instead of~\(2\). Because exchanging the two vertices of the free outer pair
of this degree-3 gadget interchanges the twisted and untwisted CFI-graphs over the base graph,
and \WL[k] can distinguish CFI-graphs from all other graphs, the resulting graph is identified
by \WL[k]. This proves Property~\ref{Property: split ows identified}.

To show Property~\ref{Property: split ows splits output pair},
we start the bijective \((k+1)\)-pebble game in position \(x_1\mapsto x_2\).
Then, Spoiler uses the usual strategy of pebbling a wall which they then
move from one side of the wall graph to the other. But because the game started in position
\(x\mapsto x'\), the two graphs the game is played on differ in a twist which will finally force
Duplicator to lose.

Now, consider again the original graph \(O^k\) without splitting the input pair.
On this graph, we can extend every winning position for Duplicator in the bijective \(k\)-pebble game
between the untwisted CFI-graph \(\CFI(B_k,0)\) and the twisted CFI-graph \(\CFI(B_k,1)\) to a position in the bijective \(k\)-pebble game
between \(O^k\) and itself which is compatible with \(x_1\mapsto x_2\).
Similarly, we can extend every winning position for Duplicator
in the bijective \(k\)-pebble game between the untwisted CFI-graph \(\CFI(B_k,0)\)
and itself to a position between \(O^k\) and itself which is compatible with \(x_1\mapsto x_1\).

We call the former positions \emph{twisted} and the latter positions \emph{trapped}.
Properties~\ref{Property: Trapped and twisted partial isos},
\ref{Property: Trapped compatible with output pair}, \ref{Property: Twisted compatible with output pair}
and \ref{Property: Subpositions of trapped and twisted}
are then immediate,
and Property~\ref{Property: Trapped and twisted winning} follows from
Lemma~\ref{Lemma: CFI vs. tw} together with Lemma~\ref{Lemma: Base graph of ows has tree-width k+1}.

Because \(v\) lies on a cycle in \(B_k\), there exists an automorphism of \(\CFI(B_k)\)
which twists both outer pairs of the gadget corresponding to \(v\).
Lifting this automorphism to \(O^k\) yields an automorphism switching \(y_1\) and \(y_2\)
whilst fixing \(x_1\) and \(x_2\).
This proves Property~\ref{Property: Switching input is trapped and twisted}.
\end{proof}
\end{lemma}
In the bijective \((k+1)\)-pebble game on \(O^k\),
we say that Duplicator follows a \emph{trapped}
or \emph{twisted strategy}  if Spoiler can never reach a non-trapped or non-twisted position respectively.
Note that Properties~\ref{Property: Trapped and twisted partial isos} and~\ref{Property: Trapped and twisted winning}
together imply that trapped and twisted strategies are winning strategies.
\newcommand{\AND}{\mathrm{AND}}
\newcommand{\OR}{\mathrm{OR}}

\subsection{From monotone circuits to graphs}
We now reduce the monotone circuit value problem MCVP to the \WL[k]-identification problem.
A monotone circuit \(M\) is a circuit consisting of input nodes, each of which has value either \(\True\) or \(\False\),
a distinguished output node, and inner nodes, which are either \(\AND\)- or \(\OR\)-nodes with two inputs each.
We write \(V(M)\) for the set of nodes of \(M\). With the monotone circuit \(M\), we can associate
the evaluation function \(\val_M\colon V(M)\to\{\True,\False\}\),
which is defined in the obvious way.
The monotone circuit value problem MCVP is the following problem:
given a monotone circuit \(M\) with output node \(c\), decide whether \(\val_M(c)=\True\).
This problem is known to be hard for polynomial time~\cite{MCVP}.

% Construction of G_M
\begin{figure}
\centering
\begin{tikzpicture}
\vertex[green]{x} {1,0};
\vertex[green]{x'}{2,0}; 

\vertex{k0}{0,1};
\vertex{k1}{1,1};
\vertex{k2}{2,1};
\vertex{k3}{3,1};

\vertex[red]{i} {2,2};
\vertex[red]{i'}{3,2};

\vertex[blue]{j} {0,2};
\vertex[blue]{j'}{1,2};

\draw (x) --  (k0) (x) --  (k1);
\draw (x') -- (k2) (x') -- (k3);

\draw (i) --  (k0) (i) --  (k2);
\draw (i') -- (k1) (i') -- (k3);

\draw (j) --  (k0) (j) --  (k3);
\draw (j') -- (k1) (j') -- (k2);

\node at (1.5,-0.5) {\(G_{\OR}\)};
\end{tikzpicture}
\hspace{2cm}
\begin{tikzpicture}
\vertex[red]{i} {2,2};
\vertex[red]{i'}{3,2};

\vertex[blue]{j} {0,2};
\vertex[blue]{j'}{1,2};

\vertex[green]{y} {1,0};
\vertex[green]{y'}{2,0}; 

\draw (i)  -- (y)  -- (j);
\draw (i') -- (y') -- (j');

\node at (1.5,-0.5) {\(G_{\AND}\)};
\end{tikzpicture}
\caption{The gadgets \(G_{\OR}\) and \(G_{\AND}\) encoding \(\OR\)- and \(\AND\)-gates.
	We call the two pairs at the top their \emph{input pairs} and their bottom pair
	their \emph{output pair}.}
\label{Figure: logic gates}
\end{figure}
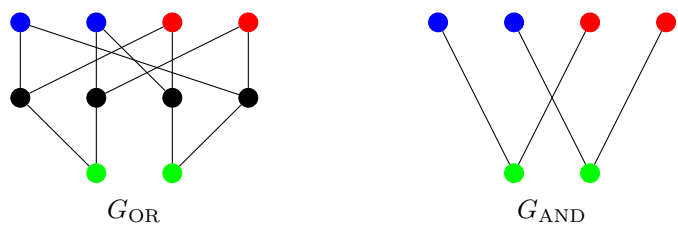

Now, let \(M\) be such a monotone circuit. We construct a colored graph \(G_M\)
such that for every node \(a\in V(M)\), there is a vertex pair \(\{a_1,a_2\}\) in \(G_M\) which will be split by \(\WL[k]\) if and only if \(\val_M(a)=\False\).
We use the two graphs \(G_{\OR}\) and \(G_{\AND}\) in Figure \ref{Figure: logic gates}
as gadgets to replace the logic gates in our construction of \(G_M\).
These gadgets both have two input pairs and one output pairs such that
exchanging the two output vertices by an automorphism
requires the two vertices of one (for \(G_{\OR}\)) or both (for \(G_{\AND}\)) input pairs
to also be exchanged.

\begin{figure}
\centering
%Usage: \splitpair{color1}{color2}{name}{center}{rotation}
\newcommand{\splitpair}[5]{
	\begin{scope}[rotate around={#5:(#4)}]
		\vertex[#1]{#3}  {$(#4)+(-0.25,0)$};
		\vertex[#2]{#3'} {$(#4)+( 0.25,0)$};
	\end{scope}
}
%Usage: \pair[color]{name}{center}{rotation}
\newcommand{\pair}[4][black]{\splitpair{#1}{#1}{#2}{#3}{#4}}

\newcommand{\connectPair}[3][]{
	\draw (#2) edge[#1] (#3) (#2') edge[#1] (#3');
}

%\Rectangle[label]{center}{rotation}{width}{height}
\newcommand{\Rectangle}[5][]{
	\begin{scope}[rotate around={#3:(#2)}]
		\draw ($(#2)-(#4/2,#5/2)$) rectangle ($(#2)+(#4/2,#5/2)$) {};
		\node at (#2) {#1};
	\end{scope}
}

%\Gate[color]{center}{name}{label}
\newcommand{\Gate}[4][black]{
	\coordinate (#3i1)  at ($(#2)+(-0.8,0.4)$);
	\coordinate (#3i1') at ($(#2)+(-0.3,0.4)$);
	\coordinate (#3i2)  at ($(#2)+( 0.3,0.4)$);
	\coordinate (#3i2') at ($(#2)+( 0.8,0.4)$);
		
	\draw ($(#2)-(0.9,0.4)$) rectangle ($(#2)+(0.9,0.4)$) {};
	\node at (#2) (a) {#4};

	\pair[#1]{#3}{$(#2)+(0,-0.4)$}{0}
}

%\OWS[rotation]{center}{x}{y}
\newcommand{\OWS}[4][0]{
	\begin{scope}[shift={(#2)}, rotate around={#1:(0,0)}]
	\coordinate (ly)  at (-0.25, 0.7);
	\coordinate (ly') at ( 0.25, 0.7);
	\coordinate (lx)  at (-0.25,-0.7);
	\coordinate (lx') at ( 0.25,-0.7);
		
	\draw (-0.35,-0.7) rectangle (0.35,0.7) {};
	\node[rotate=-90+#1] at (0,0) (a) {\(\longrightarrow\)};

	\connectPair{#3}{lx}
	\connectPair{#4}{ly}
	\end{scope}
}

\begin{tikzpicture}[scale=0.8]
\vertex[red]{i1}{0,0};
\vertex[blue]{i2}{1.5,0};
\vertex[green]{i3}{3,0};

\node[above=0.25] at (i1) {\(\False\)};
\node[above=0.25] at (i2) {\(\False\)};
\node[above=0.25] at (i3) {\(\True\)};

\node[shape=circle, draw=cyan, inner sep=2pt] (AND) at (0.75,-2.9)  {\(\land\)};
\node[shape=circle, draw=magenta, inner sep=2pt] (OR)  at (1.5,-5.8) {\(\lor\)};

\draw (i1) edge[->] (AND) (i2) edge[->] (AND) (AND) edge[->] (OR) (i3) edge[->] (OR);

\node at (1.5,-7.5) {\(M\)};

\begin{scope}[xshift=6cm]
\splitpair{red,draw=black}{orange}{i1}{0,0}{0}
\splitpair{blue,draw=black}{red!50!blue}{i2}{1.5,0}{0}
\pair[green]{i3}{3,0}{0}

\Gate[cyan]{0.75,-2.5}{and}{\(G_{\AND}\)}
\Gate[magenta]{1.5,-5.4}  {or} {\(G_{\OR}\)}

\OWS{0,  -1}{andi1}{i1}
\OWS{1.5,-1}{andi2}{i2}

\OWS{0.75,-3.9}{ori1}{and}
\OWS{2.25,-3.9}{ori2}{i3}

% backpropagation
\begin{scope}[shift={(2.75,-6.5)}, rotate around={90:(0,0)}]
	\coordinate (ly)  at (-0.25, 0.7);
	\coordinate (ly') at ( 0.25, 0.7);
	\coordinate (lx)  at (-0.25,-0.7);
	\coordinate (lx') at ( 0.25,-0.7);
		
	\draw (-0.35,-0.7) rectangle (0.35,0.7) {};
	\node at (0,0) (a) {\(\longrightarrow\)};
\end{scope}
\connectPair{ly}{or}
\connectPair[out=40, in=-45]{lx}{i3}

\node at (1.5,-7.5) {\(G_M\)};
\end{scope}
\end{tikzpicture}

\caption{A simple monotone circuit \(M\) and the graph \(G_M\) obtained from it.
	The colors in \(M\) are just for illustration purposes and not part of the actual circuit.}
\label{Figure: Graph obtained from monotone circuit}
\end{figure}
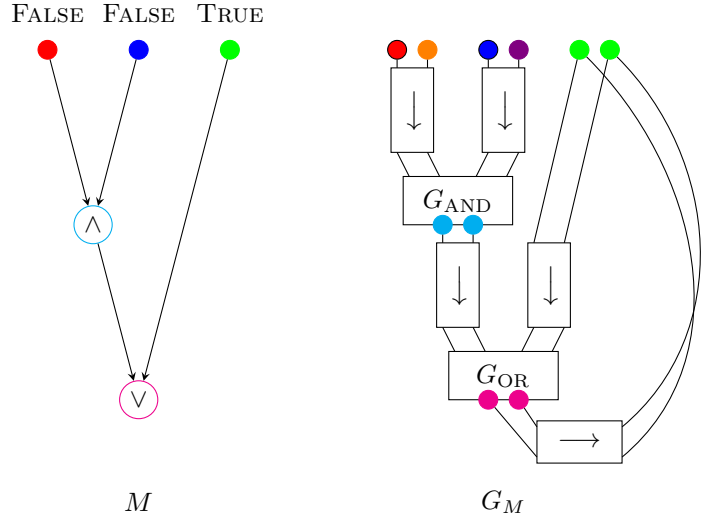

We now start with the formal construction of \(G_M\), which is depicted in Figure~\ref{Figure: Graph obtained from monotone circuit}.
For every node \(a\) of \(M\), we add a pair of vertices \(\{a_1,a_2\}\) forming a color class to \(G_M\).
To encode the input valuation of the circuit,
we split every pair \(\{a_1,a_2\}\) corresponding to an input node \(a\)
of value \False.

For every \(\AND\)-node \(a\in V(M)\) with input nodes \(b\) and \(b'\),
we add a freshly colored copy of the gadget \(G_{\AND}\) from
Figure~\ref{Figure: logic gates}
and identify its output pair with the pair \(\{a_1,a_2\}\).
Next, we connect its two input pairs via freshly colored one-way switches
\(O_{ba}^k\) and \(O_{b'a}^k\)
to the pairs \(\{b_1,b_2\}\) and \(\{b_1',b_2'\}\) respectively.
More precisely, we identify the input pairs of these one-way switches with \(\{b_1,b_2\}\)
or \(\{b_1',b_2'\}\) respectively and identify their output pair with the
respective input pair of the copy of \(G_{\AND}\).
Analogously, we add a copy of \(G_{\OR}\) for every \(\OR\)-node \(a\in V(M)\)
and connect it to its input via one-way switches as before.

This concludes the translation of the circuit itself, but for our reduction to the
identification problem we need one more step:
we connect all input pairs \(\{a_1,a_2\}\) with \(\val_M(a)=\True\) to the output pair
\(\{c_1,c_2\}\) via additional one-way switches \(O^k_{ca}\), whose input pair we identify with \(\{c_1,c_2\}\)
and whose output pair we identify with \(\{a_1,a_2\}\).

Let \(G_M\) be the resulting graph. Because we color different gadgets using distinct colors,
and every gadget has \(4\)-bounded color classes, the resulting graph \(G_M\) also has \(4\)-bounded color classes,
which could be made abelian by introducing colored edges within the gadgets.
Indeed, note that every pair of vertices in a color class of order \(4\) has precisely one common neighbor.
If we connect each such pair by an edge whose color encodes the color of the common neighbor,
the resulting color class is abelian. Note further that this coloring is computed by \WL[2],
which means that adding it does not affect the power of the Weisfeiler-Leman algorithm on these graphs.

Moreover, the graph \(G_M\) can be constructed in \(\ACZ\) because we only
need to replace every node and edge of \(M\) by a gadget of bounded size that only
depends on whether the node is an input node, an \(\AND\)-node or an \(\OR\)-node.

Recall that \(G_M\) contains, for every node \(a\) of \(M\), a vertex pair \(\{a_1,a_2\}\).
We first prove that the vertex pairs of all nodes evaluating to \False are indeed split by \WL[k].
\begin{lemma}[compare {\cite[Section 5.4]{WLEquivalencePHard}, \cite[Proof of Theorem 7.11]{CRIdentificationPHard}}]\label{Lemma: G_M splits false pairs}
For every monotone circuit \(M\) with output node \(c\), and every node \(a\) of \(M\), if either \(\val_M(a)=\False\) or \(\val_M(c)=\False\), then
\((G_M,a_1)\not\equiv_{\WL[k]}(G_M,a_2)\)
\begin{proof}
We split the proof into two claims.
\begin{claim}
If \(\val_M(a)=\False\), then \((G_M,a_1)\not\equiv_{\WL[k]}(G_M,a_2)\).
\begin{claimproof}
We argue by induction on the depth of~\(a\) in the circuit~\(M\).
If~\(a\) is an input node, then~\(a_1\) and~\(a_2\) have different colors and are thus distinguished
by construction.
Thus, assume that \(a\) is an inner node with parents \(b\) and \(b'\)
and assume by induction that the claim is true for both \(b\) and \(b'\).

If \(a\) is an \(\AND\)-node, assume w.l.o.g. that \(\val_M(b)=\False\).
By the induction hypothesis, \((G_M,b_1)\not\equiv_{\WL[k]}(G_M,b_2)\).
But this means that the input pair of the one-way switch \(O^k_{ba}\)
is split, which implies by Property~\ref{Property: split ows splits output pair}
of the one-way switch that also the output pair is split.
But this output pair is one of the input pairs of the \(\AND\)-gadget \(G_{\AND}\),
whose output pair is \(\{a_1,a_2\}\). Thus, \(\WL[k]\) also splits the pair \(\{a_1,a_2\}\).

If \(a\) is an \(\OR\)-node, then \(\val_M(b)=\val_M(b')=\False\). Thus,
by the induction hypothesis, both of the pairs \(\{b_1,b_2\}\) and \(\{b_1',b_2'\}\) are split
by \(\WL[k]\). Using Property~\ref{Property: split ows splits output pair} of the one-way switches~\(O^k_{ba}\) and~\(O^k_{b'a}\), this means
that \(\WL[k]\) also splits both input pairs of the \(\OR\)-gadget \(G_\OR\)
whose output pair is \(\{a_1,a_2\}\). However, all inner vertices of the gadget \(G_{\OR}\)
are uniquely determined by their neighborhood among the input pairs, and the two vertices of
the output pair have distinct neighborhoods among these inner vertices. Thus,
\(\WL[k]\) also splits the pair \(\{a_1,a_2\}\).
\end{claimproof}
\end{claim}

\begin{claim}
If \(\val_M(c)=\False\), then for all nodes \(a\) of \(M\) it holds that \((G_M,a_1)\not\equiv_{\WL[k]}(G_M,a_2)\).
\begin{claimproof}
Using the previous claim, we get that \WL[k] splits the pair \(\{c_1,c_2\}\).
Using the one-way switches \(O^k_{ca}\) connecting the pair \(\{c_1,c_2\}\)
to all pairs \(\{a_1,a_2\}\) corresponding to input nodes of value \True,
this means that \WL[k] now splits all pairs corresponding to input nodes
by Property~\ref{Property: split ows splits output pair} of these one-way switches.
Thus, the graph now simulates the circuit with all input nodes set to \False.
But as the circuit is monotone, this means that also all inner nodes now evaluate to \False.
By the previous claim, \WL[k] thus also splits all pairs \(\{a_1,a_2\}\) corresponding to inner nodes. 
\end{claimproof}
\end{claim}
Combining the two claims finishes the proof.
\end{proof}
\end{lemma}

\noindent Next, we prove the converse: Nodes of value \True get interpreted as vertex pairs which are not split by \(\WL[k]\).
\begin{lemma}[compare {\cite[Section 5.4]{WLEquivalencePHard}, \cite[Proof of Theorem 7.11]{CRIdentificationPHard}}]
\label{Lemma: G_M does not split true pairs}
For every monotone circuit \(M\) with output node \(c\), and every node \(a\in V(M)\), if \(\val_M(a)=\True\) and \(\val_M(c)=\True\), then
\((G_M,a_1)\equiv_{\WL[k]}(G_M,a_2)\). 
\begin{proof}
To prove the desired equivalences, we provide winning strategies for Duplicator
in the bijective \((k+1)\)-pebble game on \(G_M\).
We call a vertex pair \(\{v_1,v_2\}\) in \(G_M\) \emph{true} if it either
corresponds to a node \(v\) of \(M\) with \(\val_M(v)=\True\),
or is the output pair of a one-way switch whose input pair corresponds to a node of value \(\True\).
Otherwise, we call \(\{v_1,v_2\}\) \emph{false}.
We call a one-way switch in \(G_M\) \emph{true} if its input pair is a true pair,
and \emph{false} if its input pair is a false pair.

Now, we call a position \(P\) in the bijective \(k+1\)-pebble game on \(G_M\) \emph{safe},
if there is an automorphism \(\phi\) of the subgraph induced by all
vertices of \(G_M\) besides the inner vertices of one-way switches satisfying the following conditions:
\begin{enumerate}
\item\label{Condition: P cup phi is partial iso}
	The position \(P\cup\phi\), that is, the position obtained from \(P\)
	by adding all pebble pairs \(x\mapsto\phi(x)\) for \(x\in\dom(\phi)\),
	is a partial isomorphism.
\item\label{Condition: identity on false one-way switches} For every false one-way switch \(O\),
	\(\phi\) fixes both input and both output vertices of \(O\),
	and the restriction \(P|_O\) of the position \(P\) to vertices of \(O\) is the identity map.
\item\label{Condition: winning on every one-way switch} One of the following is true:
	\begin{enumerate}
	\item\label{Condition: all pebbles in one one-way switch} There is a single true one-way switch \(O\) with output pair \(\{x_1,x_2\}\)
		that contains all \(k+1\) pebble pairs,
		and either \(P\) is trapped and \(\phi\) fixes both \(x_1\) and \(x_2\),
		or \(P\) is twisted and \(\phi\) exchanges \(x_1\) and \(x_2\).
	\item\label{Condition: not all pebbles in one one-way switch}
		Every true one-way switch contains at most \(k\) pebble pairs, and
		for every true one-way switch \(O\) with input pair \(\{y_1,y_2\}\) and output pair \(\{x_1,x_2\}\),
		either \(P|_O\cup \{y_1\mapsto\phi(y_1)\}\) is trapped and \(\phi\) fixes both \(x_1\) and \(x_2\),
		or \(P|_O\cup \{y_1\mapsto\phi(y_1)\}\) is twisted and \(\phi\) exchanges \(x_1\) and \(x_2\).
	\end{enumerate}
\end{enumerate}
We call a position \emph{dangerous} if it is not safe.

\begin{claim}\label{Claim: Exchanging vertices for true node is safe}
For every node \(a\in V(M)\) with \(\val_M(a)=\True\), the position \(a_1\mapsto a_2\) is safe.
Moreover, the witnessing automorphism \(\phi\) can be chosen such that it is the identity
on every gadget but the gadget corresponding to the node \(a\).
\begin{claimproof}
Let $a$ be an arbitrary node of $C$.
If \(a\) is an input node, let \(\phi\) be the map that exchanges~\(a_1\) and~\(a_2\)
and fixes every other vertex. This clearly satisfies
Conditions \ref{Condition: P cup phi is partial iso} and \ref{Condition: identity on false one-way switches}.
Condition~\ref{Condition: winning on every one-way switch} is trivially true, because every color-preserving permutation of the input- and output vertices
of a one-way switch is either twisted or trapped
by Properties~\ref{Property: Trapped compatible with output pair},~\ref{Property: Twisted compatible with output pair} and~\ref{Property: Switching input is trapped and twisted} of the one-way switches.

If \(a\) is an \(\AND\)-node, pick \(\phi\) to be the unique non-trivial automorphism
on the corresponding \(\AND\)-gadget \(G_{\AND}\), and the identity everywhere else.
Again, all three conditions are clearly satisfied in this case
because both parents of \(a\) also have value \True.
If \(a\) is an \(\OR\)-node, let \(b\) be a parent node of \(a\) with \(\val_M(b)=\True\),
and let \(\phi\) be the unique automorphism that exchanges \(a_1\) and \(a_2\),
as well as \(b_1\) and \(b_2\) on the corresponding \(\OR\)-gadget \(G_{\OR}\),
and that is the identity everywhere else. Again, all three conditions are clear by construction.
\end{claimproof}
\end{claim}

\begin{claim}
Every subposition of a safe position \(P\) is safe.
\begin{claimproof}
Assume \(P\) is a safe position and let \(\phi\) be the automorphism from the definition of safeness for \(P\).

If \(P\) satisfies Condition~\ref{Condition: not all pebbles in one one-way switch},
then \(\phi\) also witnesses that every subposition of \(P\) is safe.
Indeed, because subpositions of partial isomorphisms are again partial isomorphisms, Condition~\ref{Condition: P cup phi is partial iso} is clear.
Similarly, Condition~\ref{Condition: identity on false one-way switches} is true because restrictions
of the identity are again the identity on the restricted domain.
Condition~\ref{Condition: not all pebbles in one one-way switch} follows from the observation that
every subposition of a trapped position is trapped, and every subposition of a twisted position is twisted
by Property~\ref{Property: Subpositions of trapped and twisted} of the one-way switches.

If \(P\) satisfies Condition~\ref{Condition: all pebbles in one one-way switch} instead,
then no proper subposition of \(P\) does also satisfy Condition~\ref{Condition: all pebbles in one one-way switch}.
Let \(O\) be the unique true one-way switch that contains all \(k+1\) pebble pairs,
and let \(P_0\subsetneq P\) be a proper subposition.
By Condition~\ref{Condition: all pebbles in one one-way switch} of \(P\),
the position \(P\) and thus also \(P_0\) is either trapped or twisted.
In order to prove Condition~\ref{Condition: not all pebbles in one one-way switch}
for the position \(P_0\), we need to show that we can extend this trapped or twisted position
by the vertex \(y_1\). For this, consider the bijective \((k+1)\)-pebble game in position \(P_0\)
on \(O\). Assume Spoiler picks up an unused pebble pair, and let \(\psi\colon V(O)\to V(O)\)
be the bijection that Duplicator responds with when following a trapped or twisted strategy.
Then the position \(P_0\cup\{y_1\mapsto\psi(y_1)\}\) is again trapped or twisted.
Thus, in order to ensure Condition~\ref{Condition: not all pebbles in one one-way switch} for \(P_0\),
it remains to construct the automorphism \(\phi_0\) for \(P_0\) such that \(\phi_0(y_1)=\psi(y_1)\).

To construct this automorphism, we start with the automorphism \(\phi\) that witnessed safeness
of \(P\). The position \(P_0\) together with \(\phi\) satisfies Condition~\ref{Condition: P cup phi is partial iso}.
If \(\phi(y_1)=\psi(y_1)\), we are thus done. Hence, assume from now on that \(\phi(y_1)\neq\psi(y_1)\).
Because the input pair \(\{y_1,y_2\}\) corresponds to a node of value \True,
Claim~\ref{Claim: Exchanging vertices for true node is safe} gives us a second automorphism \(\theta_y\) which exchanges
\(y_1\) and \(y_2\), and is the identity on every gadget besides the input-, \(\AND\)- or \(\OR\)-gadget containing
\(y_1\) and \(y_2\).
We set \(\phi_0\coloneqq\theta_y\circ\phi\). Then \(\phi_0\) is an automorphism satisfying
\(\phi_0(y_1)=\psi(y_1)\), which furthermore agrees with \(\phi\) everywhere but on the gadget containing \(y_1\) and \(y_2\).
Moreover, by Condition~\ref{Condition: identity on false one-way switches} applied to the position \(y_1\mapsto y_2\) and \(\theta\),
the automorphism~\(\phi_0\) disagrees with \(\phi\) only on \(y_1\) and \(y_2\), where \(P_0\cup\phi_0\) is a partial isomorphism by construction,
and possibly on output pairs of one-way switches whose input pairs correspond to a node of value \True.
But for these one-way switches, we can independently exchange or fix the two vertices of the input and the output pair
while keeping the resulting position either trapped or twisted using Properties~\ref{Property: Trapped compatible with output pair},~\ref{Property: Twisted compatible with output pair} and~\ref{Property: Switching input is trapped and twisted} of the one-way switches.
This proves Conditions~\ref{Condition: P cup phi is partial iso}
and \ref{Condition: identity on false one-way switches}.
\end{claimproof}
\end{claim}

\begin{claim}
Duplicator can avoid dangerous positions from safes ones in the bijective \((k+1)\)-pebble game
on \(G_M\).
\begin{claimproof}
Let \(P\) be a safe position and \(\phi\) the witnessing automorphism.
As subpositions of safe positions are again safe,
it suffices to assume that \(P\) contains at most \(k\) of the possible \(k+1\) pebble pairs.
Now, if a single true one-way switch \(O\) contains all \(k\) pebbles,
then the position \(P=P|_O\) is trapped or twisted,
and Duplicator picks a bijection \(\psi|_O\) according to a trapped or twisted strategy.

If \(\psi|_O\) and \(\phi\) agree on the input pair \(\{y_1,y_2\}\) of \(O\), then \(\psi|_O\) is compatible with \(\phi\).
Otherwise, we again pick the automorphism \(\theta_y\) that Claim~\ref{Claim: Exchanging vertices for true node is safe} yields
for the position \(y_1\mapsto y_2\) and replace \(\phi\) by \(\theta_y\circ\phi\), which is now compatible with \(\psi|_O\).

As all other one-way switches \(O'\) do not contain any pebble pairs on inner vertices,
and the position \(P\) is safe, the restrictions \(\phi|_{O'}\) are the identity on all
false one-way switches, and either trapped or twisted on all true one-way switches
by Properties~\ref{Property: Trapped compatible with output pair},~\ref{Property: Twisted compatible with output pair} and~\ref{Property: Switching input is trapped and twisted}.
We choose \(\psi|_{O'}\) according to a trapped, twisted or identity strategy
on these one-way switches.
Combining the maps \(\psi|_{O'}\) for all one-way switches with \(\phi\)
yields a total bijection \(\psi\), which Duplicator can choose.

Now, if Spoiler places the free pebbles on vertices \(x\) and \(\psi(x)\),
we need to argue that the resulting position is still safe.
However, Condition~\ref{Condition: P cup phi is partial iso} is true because \(\psi\)
was chosen according to a winning strategy on each gadget, and these strategies
are compatible at the intersections of the gadgets.
Conditions~\ref{Condition: identity on false one-way switches} and \ref{Condition: winning on every one-way switch} are true by construction.
Thus, the position \(P\cup\{x\mapsto\psi(x)\}\) is again safe.

Assume now that every true one-way switch of \(G_M\) contains at most \(k-1\) pebble pairs of \(P\),
and let \(\phi\) be the map which witnesses that \(P\) is safe.
Then, Condition \ref{Condition: not all pebbles in one one-way switch} guarantees that for all one-way switches \(O\) with input pair \(\{y_1,y_2\}\) and output pair \(\{x_1,x_2\}\), we can pick the bijection \(\psi|_O\)
according to a trapped, twisted or identity strategy on \(P|_O\cup\{x_1\mapsto\phi(x_1),y_1\mapsto\phi(y_1)\}\).
This ensures that \(\psi|_O\) is compatible with \(\phi\), which means that we get a total bijection \(\psi\), which Duplicator chooses.
Again, each of the positions \(P\cup\{x\mapsto\psi(x)\}\) is safe by construction.
\end{claimproof}
\end{claim}

Thus, for every node \(a\) with \(\val_M(a)=\True\),
Duplicator has a winning strategy in position \(a_1\mapsto a_2\)
by always staying in safe positions. This proves \((G_M,a_1)\equiv_{\WL[k]}(G_M,a_2)\).
\end{proof}
\end{lemma}

\noindent Similar to~\cite{WLEquivalencePHard}, combining Lemma~\ref{Lemma: G_M splits false pairs} with Lemma~\ref{Lemma: G_M does not split true pairs} yields:
\begin{corollary}\label{Corollary: WL-equivalence is P-hard}
The \WL[k]-equivalence problem for vertices is \Ptime-hard under uniform \ACZ-reductions,
both over simple graphs with \(4\)-bounded abelian color classes, and over uncolored simple graphs.
\begin{proof}
The claim for colored simple graphs with \(4\)-bounded abelian color classes
follows from \Ptime-hardness of the monotone circuit value problem,
the observation that the graph \(G_M\) can be constructed in \(\ACZ\)
together with Lemma~\ref{Lemma: G_M splits false pairs} and Lemma~\ref{Lemma: G_M splits false pairs}.

To show \(\Ptime\)-hardness of the \WL[k]-equivalence problem over uncolored simple graphs,
we need to eliminate the vertex-colors without affecting the \WL[k]-color partition
on the original graph.
To do this let \(\chi\colon G_M\to C\) be the vertex-coloring of \(G_M\),
where \(C\) is a set of colors of polynomial size. Assume w.l.o.g. that \(\chi\) is surjective,
i.e., all colors in \(C\) are used.

We construct an uncolored simple graph \(G_M^C\) from \(G_M\) as follows:
we start by adding to~\(G_M\) a fresh universal vertex \(u\).
Then, we add a path \(P\) to \(G_M\) with vertex set \(C\), and further add a fresh vertex
\(v_0\) to one end of the path. Finally, we connect each vertex~\(x\in V(G_M)\) to the vertex
\(\chi(x)\in C\). This finishes the construction of \(G_M^C\). This construction can be carried out in \(\ACZ\):
the path \(P\) is independent of the input and can hence just be hardcoded into the circuit,
and the additional edges only depend on the vertex-color.

Next, we argue that replacing the colored graph \(G_M\) by the uncolored graph \(G_M^C\)
does not affect the power of the Weisfeiler-Leman algorithm.
For this, note that the universal vertex \(u\) is the unique vertex in \(G_M^C\)
of maximal degree, and hence distinguished even by color refinement. Thus, its neighborhood,
which is the vertex set of \(G_M\subseteq G_M^C\) is also distinguished from its non-neighborhood,
which is the path \(P\). The path \(P\) contains a single vertex of degree \(1\), namely
the vertex \(v_0\) we added to one end, which means that color refinement completely discretizes
\(P\). But because the color classes of \(G_M\subseteq G_M^C\) are precisely the neighborhoods
of the vertices on \(P\) in \(G_M\), color refinement recomputes the color partition induced
on \(G_M\) by \(\chi\), while discretizing every vertex we added. The claim now follows from
the observation that adding discrete vertices which are homogeneously connected
to every color class does not affect the power of the Weisfeiler-Leman algorithm.

Thus, Lemma~\ref{Lemma: G_M splits false pairs} and Lemma~\ref{Lemma: G_M does not split true pairs}
are also true for the graph \(G_M^C\), which proves \Ptime-hardness of the \WL[k]-identification
problem for uncolored simple graphs.
\end{proof}
\end{corollary}

\noindent In order to further reduce to the identification problem, we need to consider
identification of the graph \(G_M\).
\begin{lemma}\label{Lemma: G_M is identified}
Let \(M\) be a monotone circuit with output node \(c\) which evaluates to \False.
Then \(G_M\) is identified by \WL[k].
\begin{proof}
By Lemma~\ref{Lemma: G_M splits false pairs},
all pairs \(\{a_1,a_2\}\) corresponding to a node \(a\in V(M)\) are split by \WL[k].
Translating these splits through the one-way switches, this implies that
everything in~\(G_M\) besides the one-way switches is discretized by \WL[k].
But as this includes all input and output pairs of the one-way switches,
the fact that the graphs \(O^k_{\spl}\) are identified by Property~\ref{Property: split ows identified},
implies that all one-way switches are identified as subgraphs of \(G_M\).

All individual parts of the graph \(G_M\) are thus identified as subgraphs of \(G_M\).
As moreover all edges are contained in one fo the individual parts,
and different parts only overlap in input and output pairs, all of which
are discretized by \WL[k], it follows that the whole graph~\(G_M\) is identified by \WL[k].
\end{proof}
\end{lemma}

\noindent Consider now the modified graph \(G^*_M\) we get by adding another freshly colored one-way switch \(O^k_*\) to \(G_M\)
whose input pair is \(\{c_1,c_2\}\), i.e., the vertex pair corresponding to the output node of the circuit \(M\).
Additionally, we color one of the output vertices of \(O^k_*\) in another color.
\begin{lemma}\label{Lemma: Equivalence of identification and equivalence}
The following are equivalent:
{
\renewcommand{\labelenumi}{(\roman{enumi})}
\setlength{\leftmargini}{7mm}
\begin{enumerate}
\item \label{itm:eval-false} \(\val_M(c)=\False\),
\item \label{itm:wl-distinguished} \((G_M,c_1)\not\equiv_{\WL[k]}(G_M,c_2)\) and
\item \label{itm:wl-identified} \(G^*_M\) is identified by \WL[k].
\end{enumerate}
}
\begin{proof}
The equivalence of Conditions~\ref{itm:eval-false} and~\ref{itm:wl-distinguished} follows from Lemma~\ref{Lemma: G_M splits false pairs} and Lemma~\ref{Lemma: G_M does not split true pairs}.
We show that Condition~\ref{itm:wl-distinguished} implies  Condition~\ref{itm:wl-identified}.
Assume that \((G_M,c_1)\not\equiv_{\WL[k]}(G_M,c_2)\). This splits
the input pair of \(O^k_*\). By Property~\ref{Property: split ows identified} of the one-way switches,
\(O^k_*\) is now identified as a subgraph of \(G^*_M\),
and because \(G_M\) is also identified by Lemma~\ref{Lemma: G_M is identified} and both parts
interact only at the split pair \(\{c_1,c_2\}\), the whole graph \(G^*_M\) is identified.

We finally show that Condition~\ref{itm:wl-identified} implies Condition~\ref{itm:wl-distinguished} by contraposition.
Assume \((G_M,c_1)\equiv_{\WL[k]}(G_M,c_2)\), and consider the graph \((G^*_M)'\) we obtain by
following the same construction as for \(G^*_M\), but then coloring the other output vertex
of \(O^k_*\). These two graphs are non-isomorphic, as the two output vertices of our one-way switches
do not lie in the same orbit. They are, however, equivalent under \(\WL[k]\): Duplicator can follow
a twisted strategy to win between the two one-way switches and use an arbitrary winning strategy
on \(G_M\). If none of the two parts contains all pebble pairs, these two
strategies can be made compatible at the pair \(\{c_1,c_2\}\) by enforcing which of the two
maps \(c_1\mapsto c_1\) or \(c_1\mapsto c_2\) both strategies should be compatible with.
If all pebble pairs lie in one of the two parts, Duplicator can choose their strategy on this part first
and then extend it to a winning strategy on the whole graph by using that \((G_M,c_1)\equiv_{\WL[k]}(G_M,c_2)\).
This proves that \(G^*_M\equiv_{\WL[k]} (G^*_M)'\) and thus that \(G^*_M\) is not identified by \WL[k].
\end{proof}
\end{lemma}

\noindent This finally yields our hardness result.
\IdentificationPHard*
\begin{proof}
\Ptime-hardness of the identification problem for graphs with \(4\)-bounded abelian color classes
follows from Lemma~\ref{Lemma: Equivalence of identification and equivalence}
and the hardness of the monotone circuit value problem.
Hardness for the class of simple graphs follows from the procedure
for the elimination of vertex colors which we already used in the proof of Corollary~\ref{Corollary: WL-equivalence is P-hard}.
\end{proof}

\section{Conclusion}
We have shown on the one hand that when the dimension \(k\) is part of the input,
the \WL[k]-equivalence problem and the \WL[k]-identification problem
are \coNP-hard and \NP-hard, respectively.

On the other hand, when the dimension \(k\) is fixed, the equivalence problem
is trivially solvable in polynomial time, and we have shown
that the identification problem is solvable in polynomial time
over graphs with \(5\)-bounded color classes and on relational structures with \(k\)-bounded abelian color classes.
Still, the identification problem is \Ptime-hard in both cases.
As an immediate corollary, we obtain the same polynomial-time solvability and hardness results
for definability and equivalence in the bounded-variable logic with counting \(\C^k\).

It would be interesting to know whether the \WL[k]-identification problem
can be solved in polynomial time for larger color classes or indeed on general graphs when \(k\) is fixed.
Indeed, our \NP-hardness reduction was based on whether the tree-width of a given graph is at most \(k\),
which can be solved in linear time for every fixed \(k\)~\cite{tw_linear_time},
and thus does not even yield a super-linear lower bound when \(k\) is fixed.
Still, we would expect that neither the identification nor the equivalence problem can be solved in time \(n^{o(k)}\).
It might be fruitful to study these problems from the lens of parameterized complexity
or provide lower complexity bounds based on the (strong) exponential time hypothesis.

\bibliography{references_full.bib}

\begin{thebibliography}{10}

\bibitem{dejavu}
Markus Anders and Pascal Schweitzer.
\newblock Parallel computation of combinatorial symmetries.
\newblock In {\em 29th Annual European Symposium on Algorithms, {ESA} 2021,
  September 6-8, 2021, Lisbon, Portugal (Virtual Conference)}, volume 204 of
  {\em LIPIcs}, pages 6:1--6:18. Schloss Dagstuhl - Leibniz-Zentrum f{\"{u}}r
  Informatik, 2021.
\newblock \href {https://doi.org/10.4230/LIPICS.ESA.2021.6}
  {\path{doi:10.4230/LIPICS.ESA.2021.6}}.

\bibitem{CRIdentificationPHard}
Vikraman Arvind, Johannes Köbler, Gaurav Rattan, and Oleg Verbitsky.
\newblock Graph isomorphism, color refinement, and compactness.
\newblock {\em Comput. Complex.}, 26(3):627--685, 2017.
\newblock \href {https://doi.org/10.1007/S00037-016-0147-6}
  {\path{doi:10.1007/S00037-016-0147-6}}.

\bibitem{sherali_adams1}
Albert Atserias and Elitza~N. Maneva.
\newblock {Sherali-Adams} relaxations and indistinguishability in counting
  logics.
\newblock {\em {SIAM} J. Comput.}, 42(1):112--137, 2013.
\newblock \href {https://doi.org/10.1137/120867834}
  {\path{doi:10.1137/120867834}}.

\bibitem{Quasipolynomial}
L{\'{a}}szl{\'{o}} Babai.
\newblock Graph isomorphism in quasipolynomial time [extended abstract].
\newblock In {\em Proceedings of the 48th Annual {ACM} {SIGACT} Symposium on
  Theory of Computing, {STOC} 2016, Cambridge, MA, USA, June 18-21, 2016},
  pages 684--697. {ACM}, 2016.
\newblock \href {https://doi.org/10.1145/2897518.2897542}
  {\path{doi:10.1145/2897518.2897542}}.

\bibitem{CRIdentifiesAlmostAllGraphs}
L{\'{a}}szl{\'{o}} Babai and Ludek Kucera.
\newblock Canonical labelling of graphs in linear average time.
\newblock In {\em 20th Annual Symposium on Foundations of Computer Science, San
  Juan, Puerto Rico, 29-31 October 1979}, pages 39--46. {IEEE} Computer
  Society, 1979.
\newblock \href {https://doi.org/10.1109/SFCS.1979.8}
  {\path{doi:10.1109/SFCS.1979.8}}.

\bibitem{GIBoundedCCSLasVegas}
László Babai.
\newblock {Monte}-{Carlo} algorithms in graph isomorphism testing.
\newblock Technical Report 79-10, Université de Montréal, 1979.

\bibitem{tw_linear_time}
Hans~L. Bodlaender.
\newblock A linear-time algorithm for finding tree-decompositions of small
  treewidth.
\newblock {\em {SIAM} J. Comput.}, 25(6):1305--1317, 1996.
\newblock \href {https://doi.org/10.1137/S0097539793251219}
  {\path{doi:10.1137/S0097539793251219}}.

\bibitem{bodlaender_treewidth}
Hans~L. Bodlaender.
\newblock A partial \emph{k}-arboretum of graphs with bounded treewidth.
\newblock {\em Theor. Comput. Sci.}, 209(1-2):1--45, 1998.
\newblock \href {https://doi.org/10.1016/S0304-3975(97)00228-4}
  {\path{doi:10.1016/S0304-3975(97)00228-4}}.

\bibitem{twNPhard}
Hans~L. Bodlaender, \'{E}douard Bonnet, Lars Jaffke, Du\v{s}an Knop, Paloma~T.
  Lima, Martin Milani\v{c}, Sebastian Ordyniak, Sukanya Pandey, and Ond\v{r}ej
  Such\'{y}.
\newblock Treewidth is {NP}-complete on cubic graphs.
\newblock In {\em 18th International Symposium on Parameterized and Exact
  Computation, {IPEC} 2023, September 6-8, 2023, Amsterdam, The Netherlands},
  volume 285 of {\em LIPIcs}, pages 7:1--7:13. Schloss Dagstuhl -
  Leibniz-Zentrum f{\"{u}}r Informatik, 2023.
\newblock \href {https://doi.org/10.4230/LIPICS.IPEC.2023.7}
  {\path{doi:10.4230/LIPICS.IPEC.2023.7}}.

\bibitem{2WLIdentifiesAlmostAllRegularGraphs}
B{\'e}la Bollob{\'a}s.
\newblock Distinguishing vertices of random graphs.
\newblock {\em North-holland Mathematics Studies}, 62:33--49, 1982.
\newblock \href {https://doi.org/10.1016/S0304-0208(08)73545-X}
  {\path{doi:10.1016/S0304-0208(08)73545-X}}.

\bibitem{CFI}
Jin{-}yi Cai, Martin Fürer, and Neil Immerman.
\newblock An optimal lower bound on the number of variables for graph
  identification.
\newblock {\em Comb.}, 12(4):389--410, 1992.
\newblock \href {https://doi.org/10.1007/BF01305232}
  {\path{doi:10.1007/BF01305232}}.

\bibitem{CC}
G.~Chen and I.~Ponomarenko.
\newblock {\em Lectures on Coherent Configurations}.
\newblock Central China Normal University Press, 2019.
\newblock A draft is available at \url{https://www.pdmi.ras.ru/~inp/}.

\bibitem{CFIvstw}
Anuj Dawar and David Richerby.
\newblock The power of counting logics on restricted classes of finite
  structures.
\newblock In {\em Computer Science Logic, 21st International Workshop, {CSL}
  2007, 16th Annual Conference of the EACSL, Lausanne, Switzerland, September
  11-15, 2007, Proceedings}, volume 4646 of {\em Lecture Notes in Computer
  Science}, pages 84--98. Springer, 2007.
\newblock \href {https://doi.org/10.1007/978-3-540-74915-8_10}
  {\path{doi:10.1007/978-3-540-74915-8_10}}.

\bibitem{wl_vs_homomorphisms}
Zdenek Dvor{\'{a}}k.
\newblock On recognizing graphs by numbers of homomorphisms.
\newblock {\em J. Graph Theory}, 64(4):330--342, 2010.
\newblock \href {https://doi.org/10.1002/JGT.20461}
  {\path{doi:10.1002/JGT.20461}}.

\bibitem{WL2onCCS4}
Frank Fuhlbrück, Johannes Köbler, and Oleg Verbitsky.
\newblock Identifiability of graphs with small color classes by the
  {Weisfeiler-Leman} algorithm.
\newblock {\em {SIAM} J. Discret. Math.}, 35(3):1792--1853, 2021.
\newblock \href {https://doi.org/10.1137/20M1327550}
  {\path{doi:10.1137/20M1327550}}.

\bibitem{GIBoundedCCS}
Merrick Furst, John Hopcroft, and Eugene~M. Luks.
\newblock A subexponential algorithm for trivalent graph isomorphism.
\newblock Technical report, Cornell University, USA, 1980.

\bibitem{MCVP}
Leslie~M. Goldschlager.
\newblock The monotone and planar circuit value problems are log space complete
  for {P}.
\newblock {\em {SIGACT} News}, 9(2):25--29, 1977.
\newblock \href {https://doi.org/10.1145/1008354.1008356}
  {\path{doi:10.1145/1008354.1008356}}.

\bibitem{GradelPakusa19}
Erich Gr{\"{a}}del and Wied Pakusa.
\newblock Rank logic is dead, long live rank logic!
\newblock {\em J. Symb. Log.}, 84(1):54--87, 2019.
\newblock \href {https://doi.org/10.1017/jsl.2018.33}
  {\path{doi:10.1017/jsl.2018.33}}.

\bibitem{WLEquivalencePHard}
Martin Grohe.
\newblock Equivalence in finite-variable logics is complete for polynomial
  time.
\newblock {\em Comb.}, 19(4):507--532, 1999.
\newblock \href {https://doi.org/10.1007/S004939970004}
  {\path{doi:10.1007/S004939970004}}.

\bibitem{GroheMinors}
Martin Grohe.
\newblock Fixed-point definability and polynomial time on graphs with excluded
  minors.
\newblock {\em J. {ACM}}, 59(5):27:1--27:64, 2012.
\newblock \href {https://doi.org/10.1145/2371656.2371662}
  {\path{doi:10.1145/2371656.2371662}}.

\bibitem{WLdimvstw}
Martin Grohe, Moritz Lichter, Daniel Neuen, and Pascal Schweitzer.
\newblock Compressing {CFI} graphs and lower bounds for the {Weisfeiler-Leman}
  refinements.
\newblock In {\em 64th {IEEE} Annual Symposium on Foundations of Computer
  Science, {FOCS} 2023, Santa Cruz, CA, USA, November 6-9, 2023}, pages
  798--809. {IEEE}, 2023.
\newblock \href {https://doi.org/10.1109/FOCS57990.2023.00052}
  {\path{doi:10.1109/FOCS57990.2023.00052}}.

\bibitem{BoundedTreeWidth}
Martin Grohe and Julian Mari{\~{n}}o.
\newblock Definability and descriptive complexity on databases of bounded
  tree-width.
\newblock In {\em Database Theory - {ICDT} '99, 7th International Conference,
  Jerusalem, Israel, January 10-12, 1999, Proceedings}, volume 1540 of {\em
  Lecture Notes in Computer Science}, pages 70--82. Springer, 1999.
\newblock \href {https://doi.org/10.1007/3-540-49257-7_6}
  {\path{doi:10.1007/3-540-49257-7_6}}.

\bibitem{BoundedRankWidth}
Martin Grohe and Daniel Neuen.
\newblock Canonisation and definability for graphs of bounded rank width.
\newblock {\em {ACM} Trans. Comput. Log.}, 24(1):6:1--6:31, 2023.
\newblock \href {https://doi.org/10.1145/3568025} {\path{doi:10.1145/3568025}}.

\bibitem{sherali_adams2}
Martin Grohe and Martin Otto.
\newblock Pebble games and linear equations.
\newblock {\em J. Symb. Log.}, 80(3):797--844, 2015.
\newblock \href {https://doi.org/10.1017/JSL.2015.28}
  {\path{doi:10.1017/JSL.2015.28}}.

\bibitem{Bijective_Pebble_Games}
Lauri Hella.
\newblock Logical hierarchies in {PTIME}.
\newblock {\em Inf. Comput.}, 129(1):1--19, 1996.
\newblock \href {https://doi.org/10.1006/INCO.1996.0070}
  {\path{doi:10.1006/INCO.1996.0070}}.

\bibitem{ImmermanLanderCCS3}
Neil Immerman and Eric~S. Lander.
\newblock {\em Describing Graphs: {A} First-Order Approach to Graph
  Canonization}, pages 59--81.
\newblock Springer New York, New York, NY, 1990.
\newblock \href {https://doi.org/10.1007/978-1-4612-4478-3_5}
  {\path{doi:10.1007/978-1-4612-4478-3_5}}.

\bibitem{bliss1}
Tommi~A. Junttila and Petteri Kaski.
\newblock Engineering an efficient canonical labeling tool for large and sparse
  graphs.
\newblock In {\em Proceedings of the Nine Workshop on Algorithm Engineering and
  Experiments, {ALENEX} 2007, New Orleans, Louisiana, USA, January 6, 2007}.
  {SIAM}, 2007.
\newblock \href {https://doi.org/10.1137/1.9781611972870.13}
  {\path{doi:10.1137/1.9781611972870.13}}.

\bibitem{bliss2}
Tommi~A. Junttila and Petteri Kaski.
\newblock Conflict propagation and component recursion for canonical labeling.
\newblock In {\em Theory and Practice of Algorithms in (Computer) Systems -
  First International {ICST} Conference, {TAPAS} 2011, Rome, Italy, April
  18-20, 2011. Proceedings}, volume 6595 of {\em Lecture Notes in Computer
  Science}, pages 151--162. Springer, 2011.
\newblock \href {https://doi.org/10.1007/978-3-642-19754-3_16}
  {\path{doi:10.1007/978-3-642-19754-3_16}}.

\bibitem{WL3onPlanarGraphs}
Sandra Kiefer, Ilia Ponomarenko, and Pascal Schweitzer.
\newblock The {Weisfeiler-Leman} dimension of planar graphs is at most 3.
\newblock {\em J. {ACM}}, 66(6):44:1--44:31, 2019.
\newblock \href {https://doi.org/10.1145/3333003} {\path{doi:10.1145/3333003}}.

\bibitem{CRIdentification}
Sandra Kiefer, Pascal Schweitzer, and Erkal Selman.
\newblock Graphs identified by logics with counting.
\newblock {\em {ACM} Trans. Comput. Log.}, 23(1):1:1--1:31, 2022.
\newblock \href {https://doi.org/10.1145/3417515} {\path{doi:10.1145/3417515}}.

\bibitem{KuceraRandomRegularGraphs}
Ludek Kucera.
\newblock Canonical labeling of regular graphs in linear average time.
\newblock In {\em 28th Annual Symposium on Foundations of Computer Science, Los
  Angeles, California, USA, 27-29 October 1987}, pages 271--279. {IEEE}
  Computer Society, 1987.
\newblock \href {https://doi.org/10.1109/SFCS.1987.11}
  {\path{doi:10.1109/SFCS.1987.11}}.

\bibitem{Lichter2023}
Moritz Lichter.
\newblock Separating rank logic from polynomial time.
\newblock {\em J. ACM}, 70(2), 03 2023.
\newblock \href {https://doi.org/10.1145/3572918} {\path{doi:10.1145/3572918}}.

\bibitem{Lichter2023b}
Moritz Lichter.
\newblock Witnessed symmetric choice and interpretations in fixed-point logic
  with counting.
\newblock In {\em 50th International Colloquium on Automata, Languages, and
  Programming (ICALP 2023)}, volume 261 of {\em LIPIcs}, pages 133:1--133:20.
  Schloss Dagstuhl -- Leibniz-Zentrum f{\"u}r Informatik, 2023.
\newblock \href {https://doi.org/10.4230/LIPIcs.ICALP.2023.133}
  {\path{doi:10.4230/LIPIcs.ICALP.2023.133}}.

\bibitem{nauty}
Brendan~D. McKay and Adolfo Piperno.
\newblock Practical graph isomorphism, {II}.
\newblock {\em J. Symb. Comput.}, 60:94--112, 2014.
\newblock \href {https://doi.org/10.1016/J.JSC.2013.09.003}
  {\path{doi:10.1016/J.JSC.2013.09.003}}.

\bibitem{WLgoesNeural}
Christopher Morris, Martin Ritzert, Matthias Fey, William~L. Hamilton, Jan~Eric
  Lenssen, Gaurav Rattan, and Martin Grohe.
\newblock {Weisfeiler} and {Leman} go neural: {Higher}-order graph neural
  networks.
\newblock In {\em The Thirty-Third {AAAI} Conference on Artificial
  Intelligence, {AAAI} 2019, The Thirty-First Innovative Applications of
  Artificial Intelligence Conference, {IAAI} 2019, The Ninth {AAAI} Symposium
  on Educational Advances in Artificial Intelligence, {EAAI} 2019, Honolulu,
  Hawaii, USA, January 27 - February 1, 2019}, pages 4602--4609. {AAAI} Press,
  2019.
\newblock \href {https://doi.org/10.1609/AAAI.V33I01.33014602}
  {\path{doi:10.1609/AAAI.V33I01.33014602}}.

\bibitem{multipedesII}
Daniel Neuen and Pascal Schweitzer.
\newblock Benchmark graphs for practical graph isomorphism.
\newblock In {\em 25th Annual European Symposium on Algorithms, {ESA} 2017,
  September 4-6, 2017, Vienna, Austria}, volume~87 of {\em LIPIcs}, pages
  60:1--60:14. Schloss Dagstuhl - Leibniz-Zentrum f{\"{u}}r Informatik, 2017.
\newblock \href {https://doi.org/10.4230/LIPICS.ESA.2017.60}
  {\path{doi:10.4230/LIPICS.ESA.2017.60}}.

\bibitem{multipedesI}
Daniel Neuen and Pascal Schweitzer.
\newblock An exponential lower bound for individualization-refinement
  algorithms for graph isomorphism.
\newblock In {\em Proceedings of the 50th Annual {ACM} {SIGACT} Symposium on
  Theory of Computing, {STOC} 2018, Los Angeles, CA, USA, June 25-29, 2018},
  pages 138--150. {ACM}, 2018.
\newblock \href {https://doi.org/10.1145/3188745.3188900}
  {\path{doi:10.1145/3188745.3188900}}.

\bibitem{schneider_upperBound}
Thomas Schneider and Pascal Schweitzer.
\newblock An upper bound on the {Weisfeiler-Leman} dimension, 2024.
\newblock \href {https://arxiv.org/abs/2403.12581} {\path{arXiv:2403.12581}}.

\bibitem{FibersOrder5}
Kyoungah See and Sung~Y. Song.
\newblock Association schemes of small order.
\newblock {\em Journal of Statistical Planning and Inference}, 73(1):225--271,
  1998.
\newblock \href {https://doi.org/10.1016/S0378-3758(98)00064-0}
  {\path{doi:10.1016/S0378-3758(98)00064-0}}.

\bibitem{seppelt_WLEquivcoNPhard}
Tim Seppelt.
\newblock {An Algorithmic Meta Theorem for Homomorphism Indistinguishability}.
\newblock In {\em 49th International Symposium on Mathematical Foundations of
  Computer Science (MFCS 2024)}, volume 306 of {\em Leibniz International
  Proceedings in Informatics (LIPIcs)}, pages 82:1--82:19, Dagstuhl, Germany,
  2024. Schloss Dagstuhl -- Leibniz-Zentrum f{\"u}r Informatik.
\newblock \href {https://doi.org/10.4230/LIPIcs.MFCS.2024.82}
  {\path{doi:10.4230/LIPIcs.MFCS.2024.82}}.

\bibitem{brambles}
Paul~D. Seymour and Robin Thomas.
\newblock Graph searching and a min-max theorem for tree-width.
\newblock {\em J. Comb. Theory, Ser. {B}}, 58(1):22--33, 1993.
\newblock \href {https://doi.org/10.1006/JCTB.1993.1027}
  {\path{doi:10.1006/JCTB.1993.1027}}.

\bibitem{WL}
B.~Weisfeiler and A.~Leman.
\newblock The reduction of a graph to canonical form and the algebra which
  appears therein.
\newblock {\em Nauchno-Technicheskaya Informatsia, Seriya 2}, 9:12--16, 1968.
\newblock An english translation due to Grigory Ryabov is available at
  \url{https://www.iti.zcu.cz/wl2018/pdf/wl_paper_translation.pdf}.

\bibitem{abelian_color_classes}
Faried~Abu Zaid, Erich Grädel, Martin Grohe, and Wied Pakusa.
\newblock Choiceless polynomial time on structures with small abelian colour
  classes.
\newblock In {\em Mathematical Foundations of Computer Science 2014 - 39th
  International Symposium, {MFCS} 2014, Budapest, Hungary, August 25-29, 2014.
  Proceedings, Part {I}}, volume 8634 of {\em Lecture Notes in Computer
  Science}, pages 50--62. Springer, 2014.
\newblock \href {https://doi.org/10.1007/978-3-662-44522-8_5}
  {\path{doi:10.1007/978-3-662-44522-8_5}}.

\end{thebibliography}

\end{document}